\documentclass[envcountsame,envcountreset,envcountsect]{llncs}
\usepackage{graphicx}

\usepackage{volker}
\usepackage{amsmath} 
\usepackage{amssymb} 
\usepackage{pst-node}

\usepackage{enumerate,multicol,theorem}

\newcommand{\ignore}[1]{}
\newcommand{\toappendix}[1]{\marginpar{!}{\scriptsize }}

\spnewtheorem{mydefinition}[theorem]{Definition}{\bfseries}{}

\begin{document}

\title{On the Complexity of Branching-Time Logics\thanks{Supported by
    DFG grant SCHW 678/4-1. An extended abstract of
this paper will appear in the proceedings of {\em Computer Science
  Logic (CSL) 2009}.}} 
\author{Volker Weber}
\institute{Fakult\"at f\"ur Informatik, Technische Universit\"at Dortmund\\
44221 Dortmund, Germany\\
\email{\texttt{volker.weber@tu-dortmund.de}}. }
\maketitle

\begin{abstract}
We classify the complexity of the satisfiability problem for extensions of \ctl and \ub. The extensions we consider are Boolean combinations of path formulas, fairness properties, past modalities, and forgettable past.
Our main result shows that satisfiability for \ctl with all these extensions is still in \TWOEXPTIME, which strongly contrasts with the nonelementary complexity of \ctlstar with forgettable past.
We give a complete classification of combinations of these extensions, yielding a dichotomy between extensions with \TWOEXPTIME-complete and those with \EXPTIME-complete complexity. In particular, we show that satisfiability for the extension of \ub with forgettable past is complete for \TWOEXPTIME, contradicting a claim for a stronger logic in the literature. The upper bounds are established with the help of a new kind of pebble automata.\\[2mm]\textbf{Keywords.} branching-time logic, \ctl, complexity of satisfiability, pebble automata, alternating tree automata, forgettable past
\end{abstract}

\section{Introduction}

Branching-time logics like \ctl are an important framework for the specification and verification of concurrent and reactive systems \cite{Emerson90,GrumbergV08,BaierK08}. Their history reaches almost thirty years back, when Lamport discussed the differences between linear-time and branching-time semantics of temporal logics in 1980~\cite{Lamport80}.
The first branching-time logic, called \ub, was proposed the year after by \mbox{Ben-Ari}, Pnueli, and Manna, introducing the concept of existential and universal path quantification \cite{Ben-AriPM83}. By extending \ub with the ``until'' modality, Clarke and Emerson obtained the computational tree logic \ctl \cite{ClarkeE81}, the up to date predominant branching-time logic.

Since then, many extensions of these logics have been considered. Some of these extensions aimed at more expressive power, others were introduced with the intention to make specification easier. In this paper, we consider four of these extensions that have been discussed at length in the literature, namely Boolean combinations of path formulas, fairness, past modalities, and forgettable past.

Combining these extensions, we obtain a wealth of branching-time logics. Many of the logics have been studied for their expressive power, the complexity of their satisfiability and model checking problems, and for optimal model checking algorithms. Nevertheless, for most of these logics the picture is still incomplete.

In this work, we complete the picture for the complexity of the satisfiability problem. Concretely, we completely classify the complexity of satisfiability for all branching-time logics obtained from \ub and \ctl by any combination of the extensions listed above.

Let us take a look at those parts of the picture that are already
there. The classical results in the area are the proofs of
\EXPTIME-completeness for satisfiability of \ub \cite{Ben-AriPM83} and
\ctl \cite{EmersonH85}. In the following paragraphs, we review known results for the extensions we consider.\\[1mm]
\textit{Boolean Combinations of Path Formulas.} Both, \ub and \ctl, require that every temporal operator is immediately preceded by a path quantifier. Emerson and Halpern were the first to study a logic that also allows Boolean combinations of temporal operators, i.e., of path formulas, as in $\myE(\myF p\land\neg\myF q)$ \cite{EmersonH85}. They called these logics \ubplus and \ctlplus and obtained the following hierarchy on their expressive power: $\ub\lex\ubplus\lex\ctl\eex\ctlplus$.
Concerning complexity\footnote{The \emph{complexity of a logic} always refers to the complexity of its satisfiability problem.}, \ctlplus has been shown to be complete for \TWOEXPTIME by Johannsen and Lange \cite{JohannsenL03}. The precise complexity of \ubplus is unknown.\\[1mm]
\textit{Fairness.} \ctl cannot express fairness properties, e.g., that there exists a path on which a proposition $p$ holds infinitely often. Therefore, Emerson et al. introduced \ectl by extending \ctl with a new temporal operator $\Finf$, such that $\efinf p$ expresses the property above. The logic combining \ectl with the extension discussed before, \ectlplus, roughly corresponds to the logic \ctf of \cite{EmersonC80}.

The logic \ctlstar of Emerson and Halpern extends \ectlplus with nesting of temporal operators as in $\eg(p\lor\myX p)$ \cite{EmersonH86}. Satisfiability for \ctlstar, and therefore for \ectlplus, is \TWOEXPTIME-complete \cite{VardiS85,EmersonJ99}.\\[1mm]
\textit{Past Modalities.} While being common in linguistics and philosophy, past modalities are mostly viewed only as means to make specification more intuitive in computer science. For a discussion of this issue and of the possible different semantics of past modalities, we refer to \cite{KupfermanP95,LaroussinieS00}. We adopt the view of a linear, finite, and cumulative past, which is reflected in our definition of semantics of branching-time logics based on computation trees.

We use \pctl to refer to the extension of \ctl with the past counterparts of the \ctl temporal operators, and likewise for other logics. While \pctl is strictly more expressive than \ctl \cite{KupfermanP95}, this is not the case for \pctlstar and \ctlstar \cite{HaferT87,LaroussinieS95}. In both cases, past modalities do not increase the complexity: \pctl is \EXPTIME-complete \cite{KupfermanP95} and \pctlstar has recently been shown to be \TWOEXPTIME-complete by Bozzelli \cite{Bozzelli08}.\\[1mm]
\textit{Forgettable Past.} Once past modalities are available, restricting their scope is a natural way to facilitate their use in specification. To this end, Laroussinie and Schnoebelen introduced a new operator $\myN$ for ``from now on'' to forget about the past \cite{LaroussinieS95}. I.e., past modalities in the scope of a $\myN$-operator do not reach further back than the point where the $\myN$-operator was applied. For results on the expressive power of this operator, see \cite{LaroussinieS95}.

Satisfiability for the extension of \pctl with the $\myN$-operator, \pctln, was claimed to be in \EXPTIME by Laroussinie and Schnoebelen \cite{LaroussinieS00}. In contrast to this, a nonelementary lower bound for \pctlstarn was shown in \cite{Weber09}. Nevertheless, the latter logic is known to be no more expressive than \ctlstar \cite{LaroussinieS95}.

 The logic \pectlplusn, i.e., \ctl with all the extensions considered here, also has the same expressive power as \ctlstar \cite{LaroussinieS95}. But the proof uses the separation result of Gabbay for liner temporal logic \cite{Gabbay89}, causing a nonelementary blow-up. No elementary upper bound for the complexity of satisfiability for \pectlplusn is known so far.\\[1mm]
\indent In this paper, we completely classify the complexity of satisfiability for all branching-time logics obtained from \ub and \ctl by combination of the extensions discussed above. In detail, we obtain the following results:
\begin{itemize}
	\item We show that satisfiability for all of these logics that allow Boolean combinations of path formulas is \TWOEXPTIME-complete, improving the known lower bound for \ctlplus to \ubplus.
	\item Likewise, we show that all logics with forgettable past are \TWOEXPTIME-hard, even if only the past modality $\myP$ for ``somewhere in the past'' is allowed. This contradicts the claim of Laroussinie and Schnoebelen of \EXPTIME membership for \pctln in \cite{LaroussinieS00}.
	\item We show that all logics that include neither Boolean combinations of path formulas nor forgettable past are in \EXPTIME.
	\item Finally, we show a \TWOEXPTIME upper bound for \pectlplusn, i.e.,  for \ctl with all the considered extensions. This strongly contrasts but does not contradict the nonelementary complexity of \pctlstarn, although both logics are equally expressive.
\end{itemize}

The upper bounds are obtained by translation into alternating tree automata \cite{Vardi95}. To prove the upper bound for \pectlplusn, we introduce the model of\linebreak[4] \emph{$k$-weak-pebble hesitant alternating tree automata} and show that their nonemptiness problem is in \TWOEXPTIME. These automata differ from the one-pebble alternating B\"uchi tree automata of \cite{Weber09} in two respects. First, they use a different acceptance condition to handle fairness, as proposed by Kupferman, Vardi, and Wolper for an automata model for \ctlstar model checking \cite{KupfermanVW00}. Second, the model allows more than one pebble, but only of a \emph{weak} kind. These pebbles are used to handle forgettable past and Boolean combinations of path formulas.

\subsection*{Note}
Tragically, Volker Weber died in the night of April 6-7, 2009, right after
submitting this paper to CSL 2009. An obituary is printed in the proceedings. Most of the remarks by the reviewers were incorporated by his
advisor, Thomas Schwentick. However, major changes were avoided (and
had not been requested by the reviewers). The
help of the reviewers is gratefully acknowledged.

\section{Preliminaries}

This section contains the definitions of branching-time logics and tree automata. Both are with respect to infinite trees, which we are going to define first.

A \emph{tree} is a set $T\subseteq \mathbb{N}^*$ such that if $x\cdot
c\in T$ with $x\in \mathbb{N}^*$ and $c\in \mathbb{N}$, then $x\in T$
and $x\cdot c'\in T$ for all $0<c'<c$. The empty string $\varepsilon$
is the \emph{root} of $T$ and  for all $c\in\mathbb{N}$, $x\cdot c\in
T$ is called a \emph{child} of $x$. The parent of a node $x$ is
sometimes denoted by $x\cdot-1$. We use $T_x:=\{y\in\mathbb{N}^*\mid x\cdot y\in T\}$ to denote the \emph{subtree} rooted at the node $x\in T$.
The branching degree $\degree(x)$ is the number of children of a node
$x$. Given a set $D\subseteq\mathbb{N}$, a $D$-tree is a computation
tree such that $\degree(x)\in D$ for all nodes $x$.

A \emph{path} $\pi$ in $T$ is a prefix-closed minimal set $\pi\subseteq T$, such that for every $x\in\pi$, either $x$ has no child or there is a unique $c\in \mathbb{N}$ with $x\cdot c\in\pi$. We use ``$\leq$'' (``$<$'') to denote the (strict) ancestor-relation on $T$.

A \emph{labeled tree} $(T,V)$ over a finite alphabet $\Sigma$ consists of a tree $T$ and a labeling function $V:T\rightarrow\Sigma$, assigning a symbol from $\Sigma$ to every node of $T$. We are mainly interested in the case where $\Sigma=2^{\prop}$ for some set $\prop$ of propositions. Such \emph{computation trees} result from the unfolding of Kripke structures. In the following, we consider only computation trees and refer to them as trees. We identify $(T,V)$ with $T$.

\subsection{Branching-Time Logics}

We shortly define the branching-time logics we are going to study. These definitions are mainly standard.

We start by defining the logic incorporating all the extensions discussed in the introduction. The \emph{state formulas} $\varphi$ and \emph{path formulas} $\psi$ of \pectlplusn are given by the following rules:
\begin{align*}
 \varphi & ::= p \mid \varphi\land\varphi \mid \neg\varphi \mid \myE\psi \mid \myN\varphi\\
 \psi & ::= \varphi \mid \psi\land\psi \mid \neg\psi \mid \myX\varphi \mid \varphi\myU\varphi \mid \Finf\varphi \mid \myY\varphi \mid \varphi\myS\varphi
\end{align*}
where $p\in\prop$ for some set of propositional symbols $\prop$. \pectlplusn is the set of all state formulas generated by these rules.

We use the usual abbreviations $\true,\false,\varphi\lor\varphi,\varphi\rightarrow\varphi,\varphi\leftrightarrow\varphi$, and
\begin{align*}
 \myA\psi & := \neg\myE\neg\psi & \myF\varphi & := \true\myU\varphi & \myG\varphi & := \neg\myF\neg\varphi\\
 \Ginf\varphi & := \neg\Finf\neg\varphi & \myP\varphi & := \true\myS\varphi & \myH\varphi & := \neg\myP\neg\varphi
\end{align*}

The semantics of \pectlplusn is defined with respect to a computation tree $T$, a node $x\in T$, and, in case of a path formula, a path $\pi$ in $T$ starting at the root of $T$. We omit the rules for propositions and Boolean connectives.
\begin{alignat*}{2}
 T,x & \models \myE\psi &\;&\text{iff there exists a path }\pi\text{ in }T\text{, such that }x\in\pi\text{ and }T,\pi,x\models\psi \displaybreak[1]\\
 T,x & \models \myN\varphi &&\text{iff }T_x,\varepsilon\models\varphi \displaybreak[1]\\
 T,\pi,x & \models \varphi &&\text{for a state formula }\varphi\text{, iff }  T,x\models\varphi\displaybreak[1]\\
 T,\pi,x & \models \myX\varphi &&\text{iff } T,\pi,x\cdot c\models\varphi\text{, where }c\in D\text{ and }x\cdot c\in\pi\displaybreak[1]\\
 T,\pi,x & \models \varphi_1\myU\varphi_2 &&\text{iff there is a node }y\geq x \text{ in }\pi\text{, such that }T,\pi,y\models\varphi_2\\
 & && \quad\;\text{and for all }x\leq z<y\text{ we have }T,\pi,z\models\varphi_1  \displaybreak[1]\\
 T,\pi,x & \models \Finf\varphi &&\text{iff there are infinitely many nodes }y\in\pi\text{ such that }T,\pi,y\models\varphi  \displaybreak[1]\\
 T,\pi,x & \models \myY\varphi &&\text{iff }x\neq\varepsilon\text{ and }T,\pi,x\cdot -1\models\varphi  \displaybreak[1]\\
 T,\pi,x & \models \varphi_1\myS\varphi_2 &&\text{iff there is a node }y\leq x \text{ in }\pi\text{, such that }T,\pi,y\models\varphi_2\\
 & && \quad\;\text{and for all }y<z\leq x\text{ we have }T,\pi,z\models\varphi_1  
\end{alignat*}

A formula $\varphi$ is called \emph{satisfiable} if there is a tree $T$ such that $T,x\models\varphi$.

All other logics we consider are syntactical fragments of \pectlplusn.
\begin{alignat*}{2}
 & \ubplus\quad & \varphi & := p \mid \varphi\land\varphi \mid \neg\varphi \mid \myE\psi \\
 & & \psi & := \varphi \mid \psi\land\psi \mid \neg\psi \mid \myX\varphi \mid \myF\varphi \\
 & \ubpn\; & \varphi & := p \mid \varphi\land\varphi \mid \neg\varphi \mid \ex\varphi \mid \ef\varphi \mid \af\varphi \mid \myP\varphi \mid \myN\varphi\\
 & \pectl & \varphi & := p \mid \varphi\land\varphi \mid \neg\varphi \mid \ex\varphi \mid \myE(\varphi\myU\varphi) \mid \myA(\varphi\myU\varphi)\mid \efinf\varphi \mid \myY\varphi \mid \varphi\myS\varphi\\
\end{alignat*}

\subsection{Weak-Pebble Automata}

We introduce alternating tree automata equipped with a weak kind of pebbles. We call these pebbles \emph{weak} 
 as they can only be used to mark a node while the automaton inspects the subtree below\footnote{A similar restriction on pebbles was considered in \cite{CateS08}.}. In particular, a weak-pebble automaton can only see the last pebble it dropped.

For a given set $X$, we use $\mathcal{B}^+(X)$ to denote the set of
\emph{positive Boolean formulas over $X$}, i.e., formulas built from
\true, \false and the elements of $X$ by $\wedge$ and $\vee$. A subset $Y\subseteq X$ \emph{satisfies} a boolean formula $\alpha\in\mathcal{B}^+(X)$ if and only if  assigning \true to the elements in $Y$ and \false to the elements in $X\setminus Y$ makes $\alpha$ true.

\begin{mydefinition}\label{def:kWPAA}
A \emph{$k$-weak-pebble alternating tree automaton} (\kWPAA) is a tuple $A=(Q,\Sigma,D,q^0,\delta,F)$, such that $Q$ is a finite set of states, $\Sigma$ is a finite alphabet, $D$ a finite set of arities, $q^0\in Q$ is the initial state, $F$ is an acceptance condition, and $\delta$ is a transition function 
$$\delta:Q\times\Sigma\times D\times\mathbb{B}\rightarrow  (\{\text{drop,lift}\}\times Q)\cup\mathcal{B}^+((D\cup\{-1,0,\text{root}\})\times Q)$$
such that $\delta(q,\sigma,d,\false)\neq (\text{lift},p)$, no Boolean
combination $\delta(q,\sigma,d,b)$ contains any $(d',p)$ with $d'\in D$ and $d'>d$, and no Boolean
combination  $\delta(q,\sigma,d,\true)$ contains any $(-1,p)$.
\end{mydefinition}

In the following definition of the semantics of a \kWPAA, we will use tuples $\bar{y}=(y_1,\ldots,y_k)\in(D^*\cup\{\pnp\})^k$, called pebble placements, to denote the positions of the pebbles, where ``$\pnp$'' means that the pebble is not placed. As \kWPAA will be restricted to use their pebbles in a stack-wise fashion, pebble 1 being the first pebble to be placed, there will always be an $i\in[1,k]$, such that $y_j\not=\pnp$ for all $j\leq i$ and $y_j=\pnp$ for all $j>i$. I.e., $i$ is the maximal index of a placed pebble and we will refer to it by $\mpp{\bar{y}}$. Note that $\mpp{\bar{y}}=0$ if and only if no pebble is placed and that $y_{\mpp{\bar{y}}}$ is the position of the last placed pebble otherwise.

\begin{mydefinition}
A {\em configuration} $(q,x,\bar{y})\in Q\times D^*\times(D^*\cup\{\pnp\})^k$ of a \kWPAA $A=(Q,\Sigma,D,q^0,\delta,F)$
consists of a state, the current position in the tree, and the positions of
the pebbles.

A {\em run} $r$ of $A$ on a $\Sigma$-labeled $D$-tree $(T,V)$ is a
$\mathbb{N}$-tree $(T',V')$, whose nodes are labeled by configurations
of $A$ and that is compatible with the transition function. More
precisely, the root of $T'$ must be labeled by
$(q^0,\varepsilon,\bar{y})$ with $\mpp{\bar{y}}=0$, and for every node
$v\in T'$ labeled by $(q,x,\bar{y})$ the following conditions
depending on $\delta$ hold, where $d:=\degree{x}$ and $b=\true$ if and
only if $y_{\mpp{\bar{y}}}=x$.
\begin{description}
	\item $\delta(q,V(x),d,b)=(\text{drop},p)$: If $\mpp{\bar{y}}<k$, then $v$ has a child labeled with $(q,x,\bar{y}')$, where $y_{\mpp{\bar{y}}+1}'=x$ and $y_j'=y_j$ for all $j\neq\mpp{\bar{y}}+1$. Otherwise, i.e., if all pebbles are already placed, the transition cannot be applied.
	\item $\delta(q,V(x),d,b)=(\text{lift},p)$: By Definition \ref{def:kWPAA}, $b=\true$ and therefore $y_{\mpp{\bar{y}}}=x$. Then $v$ has a child labeled with $(q,x,\bar{y}')$, where $y_{\mpp{\bar{y}}}'=\pnp$ and $y_j'=y_j$ for all $j\neq\mpp{\bar{y}}$. 
	\item $\delta(q,V(x),d,b)=\alpha$ for a boolean combination
          $\alpha\in\mathcal{B}^+((D\cup\{-1,0,\text{root}\})\times
          Q)$: There has to be a set $Y\subseteq(D\cup\{-1,0,\text{root}\})\times Q$, such that $Y$ satisfies $\alpha$,
 and, for every $(c,p)\in Y$, there is a child $v\cdot c$ of $v$ in
 $T'$ such that $v\cdot c$ is labeled by $(p,x\cdot c,\bar{y})$, where
 $x\cdot 0$ and $x\cdot \text{root}$ denote the node $x$
 itself. Additionally, we require that $Y$ does not contain a tuple
 $(-1,q)$ if $x=\varepsilon$ and that $Y$ contains a tuple
 $(\text{root},q)$ only if  $x=\varepsilon$. 
\end{description}

We call a run $r$ \emph{accepting} if every infinite path of $r$
satisfies the acceptance condition and every finite path ends in a
configuration where a transition to the Boolean combination \true applies. A labeled tree $(T,V)$ is \emph{accepted} by $A$ if and only if there is an accepting run of $A$ on $(T,V)$. The \emph{language} of $A$ is the set of trees accepted by $A$ and denoted $L(A)$.
\end{mydefinition}

\begin{mydefinition}\label{def:symkWPAA}
A \emph{symmetric $k$-weak-pebble alternating tree automaton} is a tuple $A=(Q,\Sigma,q^0,\delta,F)$, such that $Q$ is a finite set of states, $\Sigma$ is a finite alphabet, $q^0\in Q$ is the initial state, $F$ is an acceptance condition, and 
$$\delta:Q\times\Sigma\times\mathbb{B}\rightarrow(\{\text{drop,lift}\}\times Q)\cup\mathcal{B}^+((\{\mybox,\mydiamond,-1,0,\text{root}\})\times Q)$$
is a transition function such that $\delta(q,\sigma,\false)\neq (\text{lift},p)$ and $\delta(q,\sigma,\true)$ does not contain any $(-1,p)$.

The semantics of a symmetric \kWPAA are defined as for (nonsymmetric) \kWPAA, except for the last case, $\delta(q,V(x),d,b)=\alpha$, where we require that
 for every tuple $(\mydiamond,p)\in Y$ there is child of $v$ in $T'$ labeled by $(p,x\cdot c)$ for some child $x\cdot c$ of $x$ in $T$, and
 for every tuple $(\mybox,p)\in Y$ and every child $x\cdot c$ of $x$ in $T$, there is child of $v$ in $T'$ labeled by $(p,x\cdot c)$.
The conditions on tuples $(-1,p)$, $(0,p)$, and $(\text{root},p)$ remain unchanged.
\end{mydefinition}

So far, we have not specified the acceptance conditions for our automata. 
For our purposes, hesitant alternating tree automata as introduced by Kupferman, Vardi, and Wolper are a good choice \cite{KupfermanVW00}. They allow for an easy translation from \ctlstar \cite{KupfermanVW00} and have proved useful in studies of \ctlstar with past \cite{KupfermanV06,Bozzelli08}.

A \emph{$k$-weak-pebble hesitant alternating tree automaton} (\kWPHAA for short) $A=(Q,\Sigma,D,q^0,\delta,F)$ is a \kWPAA with $F=\langle G,B\rangle$, $G,B\subseteq Q$, that satisfies the following conditions:
\begin{itemize}
	\item There exists a partition of $Q$ into disjoint sets $Q_1,\ldots,Q_m$ and every set $Q_i$ is classified as either \emph{existential}, \emph{universal}, or \emph{transient}.
	\item There exists a partial order $\leq$ between the sets $Q_i$, such that every transition from a state in $Q_i$ leads to states contained either in the same set $Q_i$ or in a set $Q_j$ with $Q_j<Q_i$.
	\item If $q\in Q_i$ for a transient set $Q_i$, then $\delta(q,\sigma,d,b)$ contains no state from $Q_i$.
	\item If $Q_i$ is an existential set, $q\in Q_i$ and $\delta(q,\sigma,d,b)=\alpha$, then $\alpha$ contains only disjunctively related tuples with states from $Q_i$.
	\item If $Q_i$ is a universal set, $q\in Q_i$ and $\delta(q,\sigma,d,b)=\alpha$, then $\alpha$ contains only conjunctively related tuples with states from $Q_i$.
\end{itemize}

Every infinite path $\pi$ in a run of a \kWPHAA gets trapped in an existential or a universal set $Q_i$. The acceptance condition $\langle G,B\rangle$ is satisfied by $\pi$, if either $Q_i$ is existential and $inf(\pi)\cap G\neq\emptyset$, or $Q_i$ is universal and $inf(\pi)\cap B=\emptyset$, where $inf(\pi)$ denotes the set of states that occur infinitely often on $\pi$.

We also consider \emph{symmetric} \kWPHAA. Here, we additionally require that for every existential (resp. universal) set $Q_i$ and every state $q\in Q_i$, $\delta(q,\sigma,d,b)$ does not contain a tuple $(\Box,p)$ (resp. $(\Diamond,p)$) with $p\in Q_i$.

The size of an automaton is defined as the sum of the sizes of its components. Note that this includes the size of $D$ in the case of nonsymmetric automata.

If we remove the pebbles from our automata, we obtain \emph{(symmetric) two-way hesitant alternating tree automata}. Such an automaton has a transitions function of the form
$\delta:Q\times\Sigma\rightarrow \mathcal{B}^+((\{\mybox,\mydiamond,-1,0,\text{root}\})\times Q)$
in the symmetric case and is obtained from the above definitions in a straightforward way.

Symmetric two-way \HAA have been used by Bozzelli to prove membership in \TWOEXPTIME for \ctlstar with past \cite{Bozzelli08}. Opposed to the definition given there, we do not enforce that infinite paths in a run move only downward in the tree from a certain point on. Therefore, our results on symmetric two-way \HAA are not implied by \cite{Bozzelli08}. 
 Nevertheless, the restricted version of \cite{Bozzelli08} would suffice to prove our results on the complexity of branching-time logics.

\section{Complexity of Satisfiability}\label{sec:comp}

We give a complete classification of the complexity of the satisfiability problem for all branching-time logics obtained from \ub and \ctl by any combination of the extensions discussed above. As our main theorem, we proof \TWOEXPTIME-completeness for all logics including forgettable past or Boolean combinations of path formulas.

\begin{theorem}\label{theo:pectlplusn}
The satisfiability problems for all branching-time logics that are a syntactical fragment of \pectlplusn and that syntactically contain \ubplus or \ubpn are complete for \TWOEXPTIME.
\end{theorem}

The upper bound is proved in Section \ref{subsec:upperbound} with the help of weak-pebble alternating tree automata. There, we also show that satisfiability for all remaining logics is in \EXPTIME. 

The lower bounds on \ubpn and \ubplus are proved next.

\subsection{Lower Bounds}

 We obtain both results by reduction from the $2^n$-corridor tiling game, which is based on the \emph{$2^n$-corridor tiling problem}. An instance $I=(T,H,V,F,L,n)$ of
the latter problem consists of a finite set $T$ of tile types, horizontal and
vertical constraints $H,V\subseteq T\times T$, constraints $F,L\subseteq T$ on the first and the last row, and a number $n$ given
in unary. The task is to decide, whether $T$ tiles the
$2^n\times m$-corridor for some number $m$ of rows, respecting the
constraints.
We assume w.l.o.g. that there is always a possible next move for both players. 

The game version of this problem corresponds to alternating Turing Machines: The \emph{$2^n$-corridor tiling game} is played by two players $E$ and $A$ on an instance $I$ of the $2^n$-corridor tiling problem. The players alternately place tiles row by row starting with player $E$ and following the constraints, as the opponent wins otherwise. $E$ wins the game if a row consisting of tiles from $L$ is reached. To decide whether $E$ has a winning strategy in such a game is complete for \AEXPSPACE \cite{Chlebus86}, which is the same as \TWOEXPTIME.

\begin{proposition}\label{theo:ubpn}
Satisfiability for \ubpn is \TWOEXPTIME-hard.
\end{proposition}
\begin{proof}
Given an instance $I=(T,H,V,F,L,n)$ of the \emph{$2^n$-corridor tiling game}, we build a \ubpn-formulas $\varphi_I$ that is satisfiable if and only if player $E$ has a winning strategy on $I$.

We encode such a winning strategy as a finite $T$-labeled tree as
described in \cite{Weber09}: The levels of the tree alternately
correspond to moves of $E$ and $A$, and every node representing a move
of $E$ has one child for every next move of $A$. As each player always
has a possible move, the only way to win for $E$ is to reach a line with tiles from $L$. Therefore, every path in the encoding has to represent a tiling respecting all constraints.

The formula $\varphi_I=\varphi_s\land\varphi_n\land\varphi_t$ consists of three parts: The formula $\varphi_s$ describes the structure of the encoding and $\varphi_n$ introduces a numbering of the nodes representing one line of the tiling using the propositions $q_0,\ldots,q_{n-1}$ as shown in Figure \ref{fig:tiling}. Both are \ub-formulas and can be taken from \cite{Weber09}.
\begin{figure}[t]
\footnotesize
\begin{center}
\begin{pspicture}(0,-0.2)(10.8,1.0)
\psset{xunit=9mm}
\cnode[linecolor=red,fillstyle=solid,fillcolor=red](0,0){1mm}{a0}

\cnode(1,0){1mm}{a1}
\pnode(1.5,0){a2}
\rput(1.75,0){$\cdots$}
\pnode(2.0,0){a3}
\cnode(2.5,0){1mm}{a4}

\cnode[linecolor=red,fillstyle=solid,fillcolor=red](3.5,0){1mm}{a5}
\pnode(4.0,0){a6}
\rput(4.5,0){$\cdots$}
\pnode(5.0,0){a7}
\cnode[linecolor=red,fillstyle=solid,fillcolor=red](5.5,0){1mm}{a8}

\cnode(6.5,0){1mm}{a9}
\pnode(7.0,0){a10}
\rput(7.25,0){$\cdots$}
\pnode(7.5,0){a11}
\cnode(8,0){1mm}{a12}

\cnode[linecolor=red,fillstyle=solid,fillcolor=red](9,0){1mm}{a13}
\cnode[linecolor=black,fillstyle=solid,fillcolor=black](10,0){1mm}{a14}
\cnode[linecolor=black,fillstyle=solid,fillcolor=black](11,0){1mm}{a15}
\pnode(11.5,0){a16}
\rput(11.8,0){$\cdots$}

\ncline{->}{a0}{a1}
\ncline{->}{a1}{a2}
\ncline{->}{a3}{a4}
\ncline{->}{a4}{a5}
\ncline{->}{a5}{a6}
\ncline{->}{a7}{a8}
\ncline{->}{a8}{a9}
\ncline{->}{a9}{a10}
\ncline{->}{a11}{a12}
\ncline{->}{a12}{a13}
\ncline{->}{a13}{a14}
\ncline{->}{a14}{a15}
\ncline{->}{a15}{a16}
\psframe[linecolor=green,framearc=0.5](0.6,-0.2)(2.9,0.2)
\psframe[linecolor=green,framearc=0.5](6.1,-0.2)(8.4,0.2)

\rput(1.75,-0.4){row 1}
\rput(7.25,-0.4){row $m$}
\rput(0.0,0.4){$q_{\#}$}
\rput(3.5,0.4){$q_{\#}$}
\rput(5.5,0.4){$q_{\#}$}
\rput(9.0,0.4){$q_{\#}$}
\rput(10,0.4){$q$}
\rput(11,0.4){$q$}
\rput(1.0,0.4){$0$}
\rput(2.5,0.4){$2^{n}-1$}
\rput(6.5,0.4){$0$}
\rput(8.0,0.4){$2^{n}-1$}

\end{pspicture}
\end{center}
\caption{A path in the encoding of a winning strategy for the $2^n$-corridor tiling game with $m$ rows \cite{Weber09}.}
\label{fig:tiling}
\end{figure}
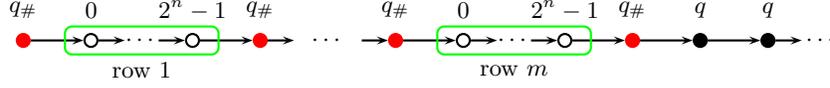

The formula $\varphi_t$ is used to describe the actual tiling. It states that every node corresponding to a position in the tiling is labeled with exactly one proposition $p_t$, representing the tile $t\in T$, and that all constraints are respected. We use $\vartheta$ as an abbreviation for $\neg q\land \neg q_{\#}\land\neg\myP(q_{\#}\land\myP(\neg q_{\#}\land \myP q_{\#}))$, expressing that the current node represents a position in the tiling and that there is at most one row above. This, together with the $\myN$-operator, will allow us to check the vertical constraints.
\begin{align*}
	\varphi_{t} & = \ag([\neg q\land \neg q_{\#}]\rightarrow[\bigvee_{t\in T}(p_t\land\bigwedge_{t\neq t'\in T}\neg p_{t'})\land\theta_H\land\theta_V\land\theta_A])\land\theta_{A'}\land\theta_F\land\theta_L\allowdisplaybreaks[1]\\
	\theta_H & = \neg\bigwedge_{i=0}^{n-1}q_i\rightarrow \bigwedge_{t\in T}(p_t\rightarrow\ax\bigvee_{(t,t')\in H} p_{t'})\allowdisplaybreaks[1]\\
	\theta_V & = \myN\ax\ag([\vartheta\land\bigwedge_{i=0}^{n-1}(q_i\leftrightarrow\myP\myH q_i)]\rightarrow\bigwedge_{t\in T}(p_t\rightarrow\bigvee_{(t,t')\in V}\myP\myH p_{t'}))\allowdisplaybreaks[1]
\end{align*}
We omit the formulas corresponding to the constraints $F$ and $L$.

The subformulas $\theta_{A}$ and $\theta_{A'}$ enforce that every possible move of $A$ is represented, where $\theta_{A'}$ treats the special case of the first row of the tiling.
\begin{align*}
	\theta_A & = q_0\rightarrow\myN\ax\ag([\vartheta\land\neg q_0\land\bigwedge_{i=1}^{n-1}(q_i\leftrightarrow\myP\myH q_i)]\\
	& \qquad\qquad\qquad\qquad\;\;\rightarrow\bigwedge_{t\in T}(p_t\rightarrow\bigwedge_{(t,t')\in H}[\ex p_{t'}\lor\myP\myH\bigvee_{(t'',t')\not\in V}p_{t''}]))\allowdisplaybreaks[1]\\
	\theta_{A'} & = \ag([\neg q_0\land\neg\myP(q_{\#}\land\myP\neg q_{\#})]\rightarrow\bigwedge_{t\in T}(p_t\rightarrow\bigwedge_{(t,t'\in H)}\ex p_{t'}))
\end{align*}

Note that a model for $\varphi_I$ might encode several possible moves for $E$ and moves of $E$ and $A$ might be represented more than once. But by removing duplicates and restricting to one arbitrary move for $E$ at every position where $E$ has to move, we obtain a winning strategy for $E$ on $I$ from a model for $\varphi_I$. For the reverse direction, a winning strategy for $E$ can be directly turned into a model for $\varphi_I$.
\qed
\end{proof}

Concerning Boolean combinations of path formulas, we can refine the proof of \TWOEXPTIME-hardness for \ctlplus by Johannsen and Lange \cite{JohannsenL03} to show the following theorem.
\begin{proposition}\label{theo:ubplus}
Satisfiability for \ubplus is \TWOEXPTIME-hard.
\end{proposition}

The proof is by reduction from the $2^n$-corridor tiling game. The main idea is to use a numbering of the rows modulo three to able to express that a path reaches up to the next row, but not beyond, without the $\myU$-operator. The details can be found in the appendix.

\subsection{Upper Bound}\label{subsec:upperbound}

We prove the upper bound of Theorem \ref{theo:pectlplusn} in two steps. First, we show how to translate a \pectl-formula into a two-way hesitant alternating tree automaton, thereby giving an exponential time algorithm for \pectl-satisfiability. Afterwards, we extend this construction to \pectlplusn and weak-pebble automata.

\begin{theorem}\label{theo:pectl}
The satisfiability problem for \pectl is \EXPTIME-complete.
\end{theorem}
\begin{proof} The lower bound follows from the \EXPTIME-hardness of \ctl. To prove the upper bound, we extend the construction from \cite{Vardi95} that translates a given \ctl-formulas into an alternating B\"uchi tree automaton.

Let $\varphi$ be a \pectl-formula and $\prop_{\varphi}$ the set of proposition symbols occurring in $\varphi$. We construct a symmetric two-way hesitant alternating tree automata $\aA_{\varphi}=(Q,\Sigma,q^0,\delta,\langle G,B\rangle)$ with $\Sigma=2^{\prop_{\varphi}}$, such that $\varphi$ holds at the root of some $\Sigma$-labeled tree $(T,V)$ if and only if $\aA_{\varphi}$ accepts this tree. The result follows since nonemptiness for these automata is in \EXPTIME (Theorem \ref{theo:TWHAA-complete}).

Let $\overline{\psi}$ denote the \emph{dual} of a formula $\psi$. The dual of a formula is obtained by switching $\land$ and $\lor$, and by negating all other maximal subformulas, identifying $\neg\neg\psi$ and $\psi$. E.g., the dual of $p\lor(\neg q\land\eg p)$ is $\neg p\land(q\lor \neg\eg p)$.

The set $Q$ of states of $\aA_{\varphi}$ is based on the Fisher-Ladner-closure of $\varphi$. It contains $\varphi$, $(\ex\efinf\psi)\land\psi$ for every subformula $\efinf\psi$ of $\varphi$, and is closed under subformulas and negation. The initial state is $\varphi$. The set $G$ contains all formulas of the form $\neg\myE(\chi\myU\psi)$, $\neg\myA(\chi\myU\psi)$, and $(\ex\efinf\psi)\land\psi$, the set $B$ all formulas of the form $\neg\ex\efinf\psi$. The transition function is defined as follows,
\begin{center}
	\begin{tabular}{rclrcll}
	  $\delta(\true,\sigma)$ & = & $\true$ & $\delta(p,\sigma)$ & = & $\true\qquad$ & if $p\in \sigma$\\
	  $\delta(\false,\sigma)$ & = & $\false$ & $\delta(p,\sigma)$ & = & $\false$ & if $p\not\in \sigma$\\
    $\delta(\psi\land\xi,\sigma)$ &  = & $(0,\psi)\land(0,\xi)\qquad\qquad$ & $\delta(\neg\psi,\sigma)$ & = & $\overline{\delta(\psi,\sigma,b)}$\\
	  $\delta(\ex\psi,\sigma)$ &  = & $(\mydiamond,\psi)$ & $\delta(\myY\psi,\sigma)$ &  = & $(-1,\psi)$\\
	\end{tabular}
	\begin{tabular}{rcll}
    $\delta(\efinf\psi,\sigma)$ & = & $(\mydiamond,\efinf\psi)\lor(0,(\ex\efinf\psi)\land\psi)$\\
    $\delta(\myE(\chi\myU\psi),\sigma)$ & = & $(0,\psi)\lor((0,\chi)\land (\mydiamond,\myE(\chi\myU\psi))$\\
    $\delta(\myA(\chi\myU\psi),\sigma)$ & = & $(0,\psi)\lor((0,\chi)\land (\mybox,\myA(\chi\myU\psi))$\\
	  $\delta(\chi\myS\psi,\sigma)$ & = & $(0,\psi)\lor((0,\chi)\land (-1,\chi\myS\psi))$
	\end{tabular}
\end{center}
where the notion of dual is extended to the transition function $\delta$ in the obvious way, e.g., $\overline{\delta(\myE(\chi\myU\psi),\sigma)} = (0,\overline{\psi})\land((0,\overline{\chi})\lor (\mybox,\neg\myE(\chi\myU\psi))$.

To show that $\aA_{\varphi}$ is a hesitant automaton, we have to define the partition of $Q$. The formulas $(\ex\efinf\psi)\land\psi$, $\ex\efinf\psi$, and $\efinf\psi$ form an existential set. Likewise, $(\neg\ex\efinf\psi)\lor\neg\psi$, $\neg\ex\efinf\psi$, and $\neg\efinf\psi$ form a universal set. Every other formula $\psi\in Q$ constitutes a singleton set $\{\psi\}$. These sets are all transient, except for the sets $\{\myE(\chi\myU\psi)\}$ and $\{\neg\myA(\chi\myU\psi)\}$, which are existential, and the sets $\{\neg\myE(\chi\myU\psi)\}$ and $\{\myA(\chi\myU\psi)\}$, which are universal. The partial order on these sets is induced by the subformula relation.
\qed
\end{proof}

We extend this construction to prove our result on \pectlplusn. To this end, we have to find a way to deal with the $\myN$-operator and Boolean combinations of path formulas. As we will see, pebbles can be used to handle both.

To obtain the desired result, it is important\footnote{Comment by
  Thomas Schwentick: this remark might puzzle the reader in the light
  of the results of Section \ref{sec:auto}. In an earlier
  version of the paper, the upper bound on the branching width in Proposition
  \ref{prop:bounded_branching} was $2^{O(n^k)}$, hence the need to
  bound the number $k$ of pebbles. However, shortly before submission time,
Volker discovered that this upper bound can be improved to
$2^{n^2\cdot(k+1)}$, thus resolving the need to bound the number of
pebbles. Seemingly, he did not find time to fully adapt (and simplify)
the paper accordingly.}

 that the number of pebbles an automaton uses does not depend on the formula from which it is constructed. But if we translate a \pectlplusn-formula into an equivalent pebble automaton along the lines of the above construction, the number of pebbles depends on the nesting depth of $\myN$-operators and Boolean combinations of path formulas. To avoid this, we show that we can restrict to formulas with limited nesting. But there is a price we have to pay: We will only obtain an equisatisfiable formula/automaton but not an equivalent one as in the proof of Theorem \ref{theo:pectl}. Nevertheless, this will suffice to obtain our complexity result.

We say that a \pectlplusn-formula is in \emph{normal form}, if it does
not contain any nesting of \myN-operators, all path quantifiers that
are followed by a Boolean combination of path formulas are not nested
and occur only in the scope of an \myN-operator, and finally all Boolean combinations of path formulas are in negation normal form.

\begin{lemma}\label{lem:normal_form}
Every \pectlplusn-formula $\varphi$ can be efficiently transformed into a \pectlplusn-formula $\psi$ in normal form with $|\psi|=O(|\varphi|)$, such that $\psi$ is satisfiable if and only if $\varphi$ is satisfiable.
\end{lemma}

We show that any \pectlplus-formula in normal form can be translated into a \kWPHAA with only two pebbles. This completes the proof of Theorem \ref{theo:pectlplusn} as nonemptiness for these automata is in \TWOEXPTIME by Theorem \ref{theo:kWPHAA-complete}.

\begin{lemma}\label{lem:logictoautomata}
Given a \pectlplusn-formula $\varphi$ in normal form, we can construct a symmetric \WPHAA{2} $\aA_{\varphi}$ of size $O(|\varphi|)$, such that $L(\aA_{\varphi})\neq\emptyset$ if and only if $\varphi\text{ is satisfiable}$.
\end{lemma}
\begin{proof}
We extend the proof of Theorem \ref{theo:pectl}, showing how to use pebbles to handle the additional features of \pectlplusn. Since $\varphi$ is given in normal form, two pebbles will suffice.

The handling of the $\myN$-operator is straightforward: We only have to drop the pebble and never lift it again. As \kWPHAA are not allowed to move above a pebble in a tree, the dropping of the pebble corresponds accurately to the meaning of the $\myN$-operator.

The handling of Boolean combinations of path formulas is more involved and mainly a matter of synchronization. An automaton corresponding to the formula $\myE(\Finf p\land\Finf q)$ cannot simply split into two automata corresponding to $\Finf p$ and $\Finf q$, respectively, as these two automata need to run on the same path in the tree. I.e., they have to be synchronized.

We will achieve synchronization using two different techniques. To this end, let $\myE\psi$ be a subformula of $\varphi$ such that $\psi$ is a Boolean combination of path formulas $\rho_1,\ldots,\rho_l$ and $\pi$ a path on which we want to evaluate $\psi$. For some of the path formulas $\rho_i$, it is the case that if they hold on $\pi$, then a finite prefix suffices to witness this, e.g., if $\rho_i=p\myU q$. For other path formulas we have to consider the whole, probably infinite path $\pi$. What thereof is the case depends on the temporal operator and whether it appears negated or not.

To evaluate $\myE\psi$ at a node $u$, $\aA_{\varphi}$ will first guess a finite prefix of a path. The intention is that this prefix can serve as a witness for all path formulas $\rho_i$ that allow for a finite witness, either to show that they hold or that they do not hold on this path.

Those parts of $\psi$ that refer to the finite prefix are now checked by $\aA_{\varphi}$ while moving up the tree again. E.g., if $\rho_i=\myF p$, then the subautomaton $\aA_{\rho_i}$ corresponding to $\rho_i$ will run upwards looking for a node $v$ labeled by $p$. As there is only one path going upward in a tree, we get synchronization for free. But we have to make sure that $v$ is a descendant of $u$. Therefore, the second pebble has to be dropped at $u$ before the automaton starts to guess the finite prefix. This allows $\aA_{\rho_i}$ to reject when it reaches the pebble position without having seen a state labeled by $p$. On the other hand, if $\rho_i=\myF\myP p$, it is important that the second pebble can be lifted again as the witness might be above $u$ in this case.

We still have to synchronize those subautomata that correspond to path formulas talking about the infinite suffix of the path. But there are only two types of conditions left, namely those of the form $\myG\chi$ and those of the form $\Finf\chi$. We can easily see that if a path satisfies a positive Boolean combination of such conditions, then every suffix of this path does so as well. This allows us to deploy the following technique.

Roughly speaking, we want to reduce the satisfiability problem for $\varphi$ to satisfiability over a restricted class of models, where the suffixes of witnessing paths for Boolean combinations of path formulas are labeled by additional propositions. More precisely, for a subformula $\myE\psi$ of $\varphi$ we introduce a new propositional symbol $p_{\psi}$ and add to $\varphi$ the new conjunct $\ag(\neg p_{\psi}\lor\eg p_{\psi})$. The automaton we are going to construct for this extended formula will, when evaluating $\myE\psi$, work as follows: It will drop the pebble and guess a prefix of a path on which $\psi$ is supposed to hold. But this prefix has to end in a node labeled by $p_\psi$. The conditions on this finite prefix can be checked as described above. For the conditions on the suffix, we use the labeling to synchronize the independent subautomata corresponding to conditions to be checked. 

Note that we cannot guaranty that there is only one path labeled by $p_{\psi}$. But we can simply check that the conditions hold on all paths labeled by $p_{\psi}$ as there is at least one such path. Please also note that this technique cannot be used for the conditions on the finite prefix of the path as the labeling is not allowed to depend on the node at which $\myE\psi$ is evaluated.
\qed
\end{proof}

\section{Nonemptiness of Weak-Pebble Automata}\label{sec:auto}

The complexity of the nonemptiness problem for weak-pebble alternating tree automata is analyzed in this section.

\begin{theorem}\label{theo:kWPHAA-complete}
The nonemptiness problem for (symmetric) $k$-weak-pebble hesitant alternating tree automata is complete for \TWOEXPTIME.
\end{theorem}

We prove this theorem for the case of nonsymmetric automata by reduction to the nonemptiness problem for Rabin tree automata \cite{Rabin69,Thomas90}. Afterwards, we generalize the result to symmetric automata by observing that every symmetric \kWPHAA accepts a tree of bounded branching degree.

\begin{mydefinition}
A \emph{nondeterministic Rabin tree automaton} (NRA) \aA is a tuple $(Q,\Sigma,D,q^0,\delta,F)$, such that $Q$ is a finite set of states, $\Sigma$ is a finite alphabet, $D$ a finite set of arities, $q^0\in Q$ is the initial state, $F$ is a set  $\{\langle G_1,B_1\rangle,\ldots,\langle G_m,B_m\rangle\}$ with $G_i,B_i\subseteq Q$, and 
$\delta:Q\times\Sigma\times D\rightarrow  2^{Q^*}$
is a transition function, such that $\delta(q,\sigma,d)\subseteq Q^d$ for all $q\in Q$, $\sigma\in\Sigma$, and $d\in D$.

We omit the (straightforward) definition of a run. An infinite path $\pi$ in a run of \aA satisfies the acceptance condition $F=\{\langle G_1,B_1\rangle,\ldots,\langle G_m,B_m\rangle\}$ if and only if there exists a pair $\langle G_i,B_i\rangle\in F$ such that $inf(\pi)\cap G_i\neq\emptyset$ and $inf(\pi)\cap B_i=\emptyset$.

A run $r$ on a tree $T$ is accepting iff every infinite path of $r$ satisfies the acceptance condition and we have $\delta(q,V(x),0)=\{\varepsilon\}$ for every finite path ending in a state $q$ at a leaf $x$ of $T$,where $\varepsilon$ denotes the sequence of states of length 0.
\end{mydefinition}

The last part of the definition is nonstandard and used to avoid the requirement that every path of $T$ is infinite. Note that \aA rejects a finite path if $\delta(q,V(x),0)=\emptyset$.

The nonemptiness problem for these automata is \NP-complete in general, but it can be solved in polynomial time if the number of tuples in the acceptance condition is bounded by a constant \cite{EmersonJ99}. For our purposes, one tuple suffices.
\begin{proposition}[\cite{EmersonJ99}
] The nonemptiness problem for nondeterministic Rabin tree automata whose acceptance condition contains only one tuple can be decided in polynomial time.
\end{proposition}

Together with the following translation, this yields Theorem \ref{theo:kWPHAA-complete} for the case of nonsymmetric \kWPHAA.
\begin{lemma}\label{lem:automata_translation}
For every \kWPHAA $\aA=(Q_{\aA},\Sigma,D,q^0_{\aA},\delta_{\aA},\langle G_{\aA},B_{\aA}\rangle)$, there is a nondeterministic Rabin tree automaton $\aB=(Q_{\aB},\Sigma,D,q^0_{\aB},\delta_{\aB},\{\langle G_{\aB},B_{\aB}\rangle\})$, such that $L(\aA)=L(\aB)$ and the number of states of \aB is at most doubly exponential in $|Q_{\aA}|$. Moreover, \aB can be constructed from \aA in exponential space.
\end{lemma}
\begin{proof}
We start with the observation that we can restrict to homogeneous runs of \aA, i.e., to runs where \aA always behaves in the same way when being in the same configuration. Formally, we call a run $r$ \emph{homogeneous}, if whenever two nodes of $r$ are labeled by the same configuration, then the set of labels occurring at there children is also the same. If \aA has an accepting run, it also has an accepting homogeneous run. This follows immediately from the existence of memoryless winning strategies for two-player parity games on infinite graphs \cite{EmersonJ91,Zielonka98}.

Next, we describe how to construct the automaton \aB from \aA. When running on a tree $T$, \aB will guess a homogeneous run $r$ of \aA and accept if and only if the guessed run $r$ is accepting. Of course, \aB cannot guess $r$ at once. Instead, \aB will guess at every node $x\in T$ how $\aA$ behaved at $x$ during $r$. The consistency of these guesses has to be checked by \aB. Additionally, \aB has to check whether the guessed run is accepting.

To perform these tasks, \aB needs to maintain some information about the guessed run $r$ of \aA. More precisely, the state taken by \aB at a node $x\in T$ will contain several sets of states of \aA that describe $r$ at $x$. Additionally, there will be some information that will be used to verify that all infinite paths in $r$ going through $x$ satisfy the acceptance condition of \aA. This information will not uniquely determine $r$, but it will be sufficient to ensure the existence of an accepting run.
A description of the information \aB stores in its states 
along with a formal definition of \aB can be found in the appendix.
\qed
\end{proof}

To transfer this result to the case of symmetric automata, we have to deal with the fact that these automata accept trees of arbitrary, even infinite branching degree. But we can show that a symmetric \kWPHAA always accepts a tree whose branching degree is at most exponential in the size of the automaton.
\begin{proposition}\label{prop:bounded_branching}
For every symmetric \kWPHAA \aA with $n$ states:\\ If $L(A)\neq\emptyset$, then \aA accepts a tree whose branching degree is at most $2^{n^2\cdot(k+1)}$.
\end{proposition}
This can be proved similarly to a corresponding result for symmetric alternating one-pebble B\"uchi automata in \cite{Weber09}. See the appendix for details.

Now, we can adapt Lemma \ref{lem:automata_translation} to symmetric \kWPHAA simply by considering every possible branching degree smaller than the bound provided by Proposition \ref{prop:bounded_branching}. This yields the upper bound of Theorem \ref{theo:kWPHAA-complete} for symmetric \kWPHAA. The matching lower bound follows from the \TWOEXPTIME-hardness of \pectlplusn via Lemma \ref{lem:logictoautomata} and Proposition \ref{prop:bounded_branching}.

Reviewing the proof of Lemma \ref{lem:automata_translation}, we observe that the resulting NRA \aB is only of exponential size if \aA does not use any pebbles, i.e., if \aA is a (symmetric) two-way HAA. This yields the following theorem, which has been proved before by Bozzelli for a slightly more restricted model \cite{Bozzelli08}.
\begin{theorem}\label{theo:TWHAA-complete}
The nonemptiness problem for (symmetric) two-way hesitant alternating tree automata is complete for \EXPTIME.
\end{theorem}

\section{Conclusions}

In this paper, we considered the branching-time logics \ub and \ctl and their extensions by Boolean combinations of path formulas, fairness, past modalities, and forgettable past. While we think that this set of extensions is a reasonable choice, there are certainly other extensions or restrictions, such as existential or universal fragments, that deserve attention.

We gave a complete classification of the complexity of the satisfiability problem for these logics, obtaining a dichotomy between \EXPTIME-complete and \TWOEXPTIME-complete logics. There are many open questions concerning the expressive power of these logics and the complexity of their model checking problems that should be addressed in future work.

\bibliographystyle{abbrv}
\bibliography{HBTL}

\appendix

\section{Comment on the disproved result from \cite{LaroussinieS00}}

As we have seen in Theorem \ref{theo:ubpn}, the satisfiability problem for \ubpn is \TWOEXPTIME-hard. This contradicts Theorem 4.1 of \cite{LaroussinieS00}, where membership in \EXPTIME is claimed for \pctln. Therefore, we want to point out where we believe the proof of \cite{LaroussinieS00} to be wrong.

A comment on notation: Our logic \pctln is called \pctl in \cite{LaroussinieS00} and $\ctl_{lp}$ is used to refer to our \pctl. In the following, we use the notation of \cite{LaroussinieS00}.

We believe that the proof of Lemma 4.3 in \cite{LaroussinieS00} is wrong. This Lemma is used to prove Theorem 4.1.

In this lemma, the extentions of \pctl and $\ctl_{lp}$ with outermost existential quantification of propositions, called \eqpctl and $\eqctl_{lp}$, are considered. The claim of Lemma 4.3 is that every \eqpctl-formula $\varphi$ can be transferred in linear time into an $\eqctl_{lp}$-formula $\tilde{\varphi}$ such that $\varphi\equiv^*\tilde{\varphi}$. ($\equiv^*$ refers to equivalence over acyclic structures, see \cite{LaroussinieS00} for details.) The authors first argue, that for every subformula $\myN\psi$ where $\psi$ contains no $\myN$, there is a \ctl-formula $\bar{\psi}$ and some new propositions $p_1',\ldots,p_k'$ such that 
\begin{equation}
\myN\psi\equiv^*\exists p_1'\cdots\exists p_k'\bar{\psi}.\label{eq:a}
\end{equation}
So far, the proof appears to be correct.

But next, they claim that they can use another proposition $p_{\bar{\psi}}'$ to obtain the following euivalence:
\begin{equation}
\varphi[\myN\psi] \equiv^* \exists p_1'\cdots\exists p_k'\exists p_{\bar{\psi}}'(\varphi[p_{\bar{\psi}}']\land\ag(p_{\bar{\psi}}'\iff\bar{\psi}))\label{eq:b}
\end{equation}

To see why Equation~\eqref{eq:b} is wrong, let us go a step back. We can use Equation~\ref{eq:a} to substitute $\myN\psi$ in $\varphi$:
\begin{equation}
\varphi[\myN\psi] \equiv^* \varphi[\exists p_1'\cdots\exists p_k'\bar{\psi}]\label{eq:c}
\end{equation}
Note that the formula on the right of Equation~\eqref{eq:c} is not a \eqpctl-formula as the quantifiers appear inside the formula. This seems to be the reason why the proposition $p_{\bar{\psi}}'$ is introduced. But this should result in the following equation, where the right-hand side is again no \eqpctl-formula.
\begin{equation}
\varphi[\myN\psi] \equiv^* \exists p_{\bar{\psi}}'(\varphi[p_{\bar{\psi}}']\land\ag(p_{\bar{\psi}}'\iff\exists p_1'\cdots\exists p_k'\bar{\psi}))\label{eq:d}
\end{equation}
But the formula on the right-hand  side of Equation~\eqref{eq:d} is clearly not equivalent to the one on the right-hand side of Equation~\eqref{eq:b}.

To put this more intuitive, note that the new propositions in
Equation~\eqref{eq:a} are used to mark those states where certain
subformulas of $\psi$ including past operators hold (see lemma 4.3 of
\cite{LaroussinieS00}). This information clearly depends on where the
preceding $\myN$-operator was used. E.g., whether $\myY\true$ holds at
a state clearly depends on whether $\myN$ was applied to this state before. This dependency is reflected in Equation~\eqref{eq:d} by newly quantifying the propositions $p_i'$ at every state when establishing the equivalence between $p_{\bar{\psi}}'$ and $\bar{\psi}$. I.e., the quantifiers are inside the $\ag$-subformula. Drawing them out, as done in \cite{LaroussinieS00}, corresponds to ignoring the dependency described above.

\section{Proof of Proposition \ref{theo:ubpn}}

For sake of completeness, we give the omitted formulas.
\begin{align*}
	\varphi_{s} & = q_{\#}\land \ag(\neg(q\land q_{\#})) \land\af(q_{\#}\land\ax q)\land\ag(q\rightarrow\ag q) \land\ex(\neg q)\\
	 \varphi_{n} & = \ag([q_{\#}\rightarrow(\ax(\neg q\land\neg q_{\#}\land\bigwedge_{i=0}^{n-1}\neg q_i)\lor\ax q)]\land\\
	 &\qquad\quad\;\:[(\neg q\land \neg q_{\#})\rightarrow((\bigwedge_{i=0}^{n-1}q_i\land\ax q_{\#}) \lor(\ax(\neg q\land \neg q_{\#})\land\nu))])\\
	 \nu & = \bigvee_{i=0}^{n-1}(\bigwedge_{j<i}(q_j\land\ax\neg q_j)\land \neg q_i \land \ax q_i\land\\
	 &\qquad\quad\,\bigwedge_{j>i}((q_j\rightarrow\ax q_j)\land (\neg q_j\rightarrow\ax \neg q_j)))\\
	\theta_{F} & = \ax\af(q_{\#}\land\myH(q_{\#}\lor \bigvee_{t\in F}p_t))\\
	\theta_{L} & = \af(q_{\#}\land\neg\ax q\land\ag(q_{\#}\lor q\lor \bigvee_{t\in L}p_t))
\end{align*}

\section{Proof of Proposition \ref{theo:ubplus}}
Given an instance $I=(T,H,V,F,L,n)$ of the \emph{$2^n$-corridor tiling game}, we build a \ubpn-formulas $\varphi_I$ that is satisfiable if and only if player $E$ has a winning strategy on $I$.

We encode the winning strategy as in the proof of Theorem \ref{theo:ubpn}, but for two extensions. First, as in \cite{JohannsenL03}, every state corresponding to a position in the tiling will have a \emph{copy-state} as child, i.e., a state where the same propositions hold except for a new proposition $c$ used to label the copy states. Furthermore, all children of copy-states are copy-states carrying the same information. Second, we use the symbols $e_1,e_2,e_3$ to introduce a numbering modulo three of the rows.

Using the following abbreviations,
\begin{align*}
	e & = e_1\lor e_2\lor e_3\\
	\psi_e & = (e_1\leftrightarrow\myX e_1)\land(e_2\leftrightarrow\myX e_2)\land(e_3\leftrightarrow\myX e_3)\\
	\psi_1 & = \myF c\land ((e_1\land \myF e_2\land\neg \myF e_3)\lor(e_2\land \myF e_3\land\neg \myF e_1)\lor(e_3\land \myF e_1\land\neg \myF e_2))
\end{align*}
we obtain the formula $\varphi_I$, where the purpose of each subformula $\theta_X$ is the same as in the proof Theorem \ref{theo:ubpn}.

\begin{align*}
\varphi_I & = \bigwedge_{i=1}^{7}\varphi_i\land\ag(e\rightarrow(\theta_H\land\theta_V\land\theta_A))\land\theta_{A'}\land\theta_F\land\theta_L\displaybreak[1]\\
\varphi_1 & = q_{\#}\land\ex \true\land \ax e_1\displaybreak[1]\\
\varphi_2 & = \ag(( q_{\#}\land\neg q\land\neg c\land\neg e_1\land\neg e_2\land \neg e_3)\\
 & \qquad\;\lor(\neg q_{\#}\land q\land\neg c\land\neg e_1\land\neg e_2\land \neg e_3)\\
 & \qquad\;\lor(\neg q_{\#}\land\neg q\land c\land\neg e_1\land\neg e_2\land \neg e_3)\\
 & \qquad\;\lor(\neg q_{\#}\land\neg q\land\neg c\land e_1\land\neg e_2\land \neg e_3)\\
 & \qquad\;\lor(\neg q_{\#}\land\neg q\land\neg c\land\neg e_1\land e_2\land \neg e_3)\\
 & \qquad\;\lor(\neg q_{\#}\land\neg q\land\neg c\land\neg e_1\land\neg e_2\land e_3))\displaybreak[1]\\
\varphi_3 & = \ag(q_{\#}\rightarrow(\ax q\lor \ax(e\land\bigwedge_{i=0}^{n-1}\neg q_i)))\land\ag(q\rightarrow\ag q)\land\\
 & \quad\;\,\ag(c\rightarrow\ag c)\displaybreak[1]\\
\varphi_4 & = \ag(e\rightarrow[\ex(e\lor q_{\#})\land\ex c\land \bigvee_{t\in T}(p_t\land\bigwedge_{t'\in T,\: t\neq t'}\neg p_{t'})])\displaybreak[1]\\
\varphi_5 & = \ag((e\lor c)\rightarrow\myA(\myX c\rightarrow[\psi_e\land\bigwedge_{i=0}^{n-1}(q_i\leftrightarrow\myX q_i)\land\bigwedge_{t\in T}(p_t\leftrightarrow \myX p_t)]))\displaybreak[1]\\
\varphi_6 & = \ag((e\land\neg\bigwedge_{i=0}^{n-1}q_i)\rightarrow\myA(\myX e\rightarrow[\psi_e\land\bigvee_{i=0}^{n-1}(\bigwedge_{j<i}(q_j\land\myX\neg q_j)\land\neg q_i\land \myX q_i\land\\
 & \qquad\qquad\qquad\qquad\qquad\qquad\qquad\qquad\qquad\quad\bigwedge_{j>i}(q_j\leftrightarrow\myX q_j))]))\displaybreak[1]\\
\varphi_7 & = \ag((e\land\bigwedge_{i=0}^{n-1}q_i)\rightarrow\myA[\myX(c\lor q_{\#})\land(e_1\rightarrow\ax\ax(c\lor q\lor e_2))\\
 & \qquad\qquad\qquad\land(e_2\rightarrow\ax\ax(c\lor q\lor e_3))\land(e_3\rightarrow\ax\ax(c\lor q\lor e_1))])\displaybreak[1]\\
\theta_H & = \bigwedge_{t\in T}(p_t\rightarrow\ax(e\rightarrow\bigvee_{(t,t')\in H}p_{t'}))\displaybreak[1]\\
\theta_V & = \myA([\psi_1\land\bigwedge_{i=0}^{n-1}(q_i\leftrightarrow\myF(c\land q_i))]\rightarrow\bigwedge_{t\in T}[p_t\rightarrow\bigvee_{(t,t')\in V}\myF(c\land p_{t'})])\displaybreak[1]\\
\theta_A & = (q_0\land\ef(q_{\#}\land\ex\neg q))\rightarrow\bigwedge_{(t,t')\in V}(p_t\rightarrow[\\
 & \qquad \myE(\psi_1\land\bigwedge_{i=0}^{n-1}(q_i\leftrightarrow\myF(c\land q_i))\land\myF(c\land p_{t'}))\lor\\
 & \qquad \myA([\psi_1\land\myF(c\land \neg q_0)\land\bigwedge_{i=1}^{n-1}(q_i\leftrightarrow\myF(c\land q_i))]\rightarrow\myF[c\land\bigvee_{(t'',t')\not\in H}p_{t''}])])\displaybreak[1]\\
 \theta_{A'} & = \myA(\neg\myF q_{\#}\rightarrow\myG((e\land\neg q_0)\rightarrow\bigwedge_{t\in T}(p_t\rightarrow\bigwedge_{(t,t')\in H}(\ex p_{t'}))))\displaybreak[1]\\
\theta_F & = \myA(\neg\myF e_2\rightarrow\myG(q_{\#}\lor\bigvee_{t\in F}p_t))\displaybreak[1]\\
\theta_L & = \ag(q_{\#}\rightarrow\myA([\myX e\land\myF q\land((e_1\land\neg\myF e_2)\lor(e_2\land\neg\myF e_3)\lor(e_3\land\neg\myF e_1))]\\
 & \quad\qquad\qquad\qquad\;\rightarrow\myG[q_{\#}\lor q\lor \bigvee_{t\in L}p_t]))
\end{align*}

\section{Missing Details from Section \ref{subsec:upperbound}}

This section contains the missing details from Section \ref{subsec:upperbound}: A formal definition of the normal form and the proofs of Lemma \ref{lem:normal_form} and Lemma \ref{lem:logictoautomata}.

\begin{definition}\label{def:nf}
A \pectlplusn-formula is in \emph{normal form}, if it can be generated by the following context free grammar.
\begin{eqnarray*}
 S_1 & \rightarrow & P \mid S_1\land S_1 \mid \neg S_1 \mid \ex S_1 \mid \myE(S_1 \myU S_1) \mid \myA(S_1 \myU S_1) \mid \\
 & & \efinf S_1 \mid \myY S_1 \mid S_1\myS S_1 \mid \myN S_2\\
 S_2 & \rightarrow & P \mid S_2\land S_2 \mid \neg S_2 \mid \ex S_2 \mid \myE(S_2 \myU S_2) \mid \myA(S_2 \myU S_2) \mid \\
 & & \efinf S_2 \mid \myY S_2 \mid S_2\myS S_2 \mid \myE S_3 \mid \myA S_3\\
 S_3 & \rightarrow & S_3\land S_3 \mid S_3\lor S_3 \mid \myX S_4 \mid \Finf S_4 \mid S_4 \myU S_4 \mid \myY S_4 \mid S_4 \myS S_4\mid\\
  &  &  \neg\myX S_4 \mid \neg\Finf S_4 \mid \neg(S_4 \myU S_4) \mid \neg\myY S_4 \mid \neg(S_4 \myS S_4)\\
 S_4 & \rightarrow & P \mid S_4\land S_4 \mid \neg S_4 \mid \ex S_4 \mid \myE(S_4 \myU S_4) \mid \myA(S_4 \myU S_4) \mid \\
 & & \efinf S_4 \mid \myY S_4 \mid S_4\myS S_4\\
 P & \rightarrow & p\qquad\text{for all }p\in\prop
\end{eqnarray*}
\end{definition}

Please observe that a formula in normal form does not contain any nesting of \myN-operators. Furthermore, path quantifiers that are followed by a Boolean combination of path formulas are not nested and such path quantifiers occur only in the scope of an \myN-operator. Moreover, Boolean combinations of path formulas are in negation normal form.

\begin{proof}[of Lemma \ref{lem:normal_form}]
The proof is by standard renaming techniques.

Let $\myN\psi$ be a subformula of $\varphi$ such that $\psi$ contains no further $\myN$-operator. We introduce a new proposition $p_{\myN\psi}$ to mark those states where $\myN\psi$ holds. Now, we replace $\myN\psi$ with $p_{\text{N}\psi}$ in $\varphi$ and add a conjunct $\ag(p_{\text{N}\psi}\leftrightarrow\myN\psi)$. The resulting formula is satisfiable if and only if $\varphi$ is satisfiable.
Repeating this for every subformula of the form $\myN\psi$ results in a formula $\varphi'\land\bigwedge\psi_i$ without nested occurrences of the $\myN$-operator.

After prefixing $\varphi'$ with a $\myN$-operator, every Boolean combination of path formulas is in the scope of some $\myN$-operator. To avoid the nesting of these Boolean combinations, we apply the same renaming technique, but inside the subformulas of the form $\myN\psi$. I.e., we substitute in $\psi$ and add conjuncts to $\psi$.

Finally, we use De Morgan's law to obtain a formula in normal form that is satisfiable if and only if $\varphi$ is satisfiable.
\qed
\end{proof}

\newpage
\begin{proof}[of Lemma \ref{lem:logictoautomata}]

We give the detailed construction along the lines described above.

As $\varphi$ is given in normal form, it can be generate by the context free grammar given in Definition \ref{def:nf}. Assuming that the rules $S_2\rightarrow\myE S_3$ and $S_2\rightarrow\myA S_3$ are used only if necessary, we call every subformula of $\varphi$ that is generated from $S_3$ a path formula and every other subformula a state formula.

Let $\psi_1,\ldots,\psi_m$ be the maximal path formulas such that $\myE\psi_j$ is a subformula of $\varphi$ and $\psi_{m+1},\ldots,\psi_{m+n}$ those where $\myA\psi_j$ is a subformula of $\varphi$. Furthermore, for every $\psi_j$ with $j\in[1,m+n]$ let $\psi_{j+m+n}$ denote the dual of $\psi_j$. We introduce for each path formula $\psi_j$ with $j\in[1,2m+2n]$ a new propositional symbol $p_j$. For every $\psi_j$, we add a conjunct to $\varphi$ to ensure that these propositions are used to label paths.
$$\varphi':=\varphi\land\ag(\bigwedge_{j\in[1,m]\cup[2m+n+1,2m+2n]}(\neg p_j\lor\eg p_j)\land\bigwedge_{j\in[m+1,2m+n]}(\neg p_j\lor\ag p_j))$$

Now, we construct a \WPHAA{2} $\aA_{\varphi'}$. The states of this automaton will not only consist of subformulas of $\varphi'$, but carry some additional information. First, $\aA_{\varphi'}$ will work in eight stages and remember in its state in which stage it is. Therefore, $\aA_{\varphi'}$ will have a state $[\psi,i]$ for every state subformula $\psi$ of $\varphi$ and every $i\in[1,8]$. When handling a Boolean combination of path formulas $\psi_j$, $\aA_{\varphi'}$ will also need to remember $j$ as it has to follow the path labeled by $p_j$. To this end, $\aA_{\varphi'}$ has states $[\psi,i,j]$, where $\psi$ is a path subformula of $\varphi$, $i\in[3,6]$, and $j\in[1,2m+2n]$.

The automaton will be in stage 1 until it handles a subformula of the form $\myN\psi$, when it will go into stage 2. I.e., $\aA_{\varphi'}$ remembers that it dropped the pebble. Stage 3 is entered when the second pebble is dropped to guess the finite prefix of a path after a subformula of the form $\myE\psi_j$ occurred. Afterwards, $\psi_j$ is evaluated in stage 4. The stages 5 and 6 serve for the same purposes, but for subformulas of the form $\myA\psi$. As soon as $\aA_{\varphi'}$  proceeds to s state subformula of $\psi$, it enters stage 7. Finally, stage 8 is entered when the second pebble is lifted.

The formulas occurring in the states of $\aA_{\varphi'}$ are not only the subformulas of $\varphi'$ and their duals. As in the proof of Theorem \ref{theo:pectl}, we need some additional formulas. E.g, for every path formula of the form $\neg(\chi\myU\psi)$, we will also need the path formulas $\myH\neg\psi$, $\myG\neg\psi$, and $\myP(\neg\chi\land)$. The actual set of states of $\aA_{\varphi'}$ can be inferred from the definition of the transition function below.

$\aA_{\varphi'}$ has the following transition rules for each $\sigma\in\Sigma=2^{\prop_{\varphi'}}$ and each $b\in\mathbb{B}$.
For $i\in\{1,2,4,6,7,8\}$:\allowdisplaybreaks[1]
\begin{align*}
	  \delta([\true,i],\sigma,b) & = \true\\
	  \delta([\false,i],\sigma,b) & = \false\\
    \delta([p,i],\sigma,b) & = \true && \mspace{-68mu}\text{if }p\in \sigma\\
    \delta([p,i],\sigma,b) & = \false && \mspace{-68mu}\text{if }p\not\in \sigma\\
    \delta([\psi\land\xi,i],\sigma,b) & = (0,[\psi,i])\land(0,[\xi,i])\\
    \delta([\psi\lor\xi,i],\sigma,b) & = (0,[\psi,i])\lor(0,[\xi,i])\\
\intertext{For $i\in\{1,2,7,8\}$:}
	  \delta([\neg\psi,i],\sigma,b) & = \overline{\delta([\psi,i],\sigma,b)}\\
	  \delta([\ex\psi,i],\sigma,b) &  = (\mydiamond,[\psi,i])\\
    \delta([\efinf\psi,i],\sigma,b) & = (\mydiamond,[\efinf\psi,i])\vee(0,[(\ex\efinf\psi)\wedge\psi,i])\\
    \delta([\myE(\chi\myU\psi),i],\sigma,b) & = (0,[\psi,i])\vee((0,[\chi,i])\wedge (\mydiamond,[\myE(\chi\myU\psi),i])\\
    \delta([\myA(\chi\myU\psi),i],\sigma,b) & = (0,[\psi,i])\vee((0,[\chi,i])\wedge (\mybox,[\myA(\chi\myU\psi),i])\\
\intertext{For $i\in\{1,2,8\}$:}
	  \delta([\myY\psi,i],\sigma,\true) & = \false \\
	  \delta([\myY\psi,i],\sigma,\false) & = (-1,[\psi,i]) \\
	  \delta([\chi\myS\psi,i],\sigma,\true) & = (0,[\psi,i])\\
	  \delta([\chi\myS\psi,i],\sigma,\false) & = (0,[\psi,i])\vee((0,[\chi,i])\wedge (-1,[\chi\myS\psi,i]))\\
\intertext{In stage 7, past operators have to be handled differently as we might have to lift the second pebble.}
	  \delta([\myY\psi,7],\sigma,\true) & = (\text{lift},[\myY\psi,8]) \\
	  \delta([\myY\psi,7],\sigma,\false) & = (-1,[\psi,7]) \\
    \delta([\chi\myS\psi,7],\sigma,\true) & = (\text{lift},[\chi\myS\psi,8])\\
    \delta([\chi\myS\psi,7],\sigma,\false) & = (0,[\psi,7])\vee((0,[\chi,7])\wedge (-1,[\chi\myS\psi,7]))\\
\intertext{In the following rules are used to switch from stage 1 to stage 2, from stage 2 to stage 3 or 5, and to guess a finite prefix of a path in stage 4 or 5 and to switch to stage 4 or 6, respectively.}
    \delta([\myN\psi,1],\sigma,\false) & = (\text{drop},[\psi,2])\\
	  \delta([\myE\psi_j,2],\sigma,b) &  = (0,[\psi_j,3,j]) \\
	  \delta([\neg\myE\psi_j,2],\sigma,b) & = (0,[\psi_{j+m+n},5,j+m+n]) \\
	  \delta([\psi_j,3,j],\sigma,b) & = (\mydiamond,[\psi_j,3,j]) && \mspace{-68mu}\text{if }p_{j}\not\in\sigma\\
	  \delta([\psi_j,3,j],\sigma,b) & = (0,[\psi_j,4,j])\vee(\mydiamond,[\psi_j,3,j]) && \mspace{-68mu}\text{if }p_{j}\in\sigma\\
	  \delta([\myA\psi_j,2],\sigma,b) & = (0,[\psi_j,5,j]) \\
	  \delta([\neg\myA\psi_j,2],\sigma,b) & = (0,[\psi_{j+m+n},3,j+m+n]) \\
	  \delta([\psi_j,5,j],\sigma,b) & = (\mybox,[\psi_j,5,j]) && \mspace{-68mu}\text{if } p_{j}\not\in\sigma\\
	  \delta([\psi_j,5,j],\sigma,b) & = (0,[\psi_j,6,j])\vee(\mybox,[\psi_j,5,j]) && \mspace{-68mu}\text{if }p_{j}\in\sigma\\
\intertext{To evaluate the Boolean combinations in stage 4 and 6, we use the following rules for $j\in\{4,6\}$. The additional state $q$ is used to check that the pebble is placed at the parent node of the current node.}
	  \delta([\myX\psi,i,j],\sigma,\false) & = ((0,[\psi,7])\land(-1,q))\lor(-1,[\myX\psi,i,j]) \\
	  \delta([\myX\psi,i,j],\sigma,\true) & = \false \\
	  \delta([\neg\myX\psi,i,j],\sigma,\false) & = ((0,[\neg\psi,7])\land(-1,q))\lor(-1,[\neg\myX\psi,i,j]) \\
	  \delta([\neg\myX\psi,i,j],\sigma,\true) & = \false \\
	  \delta(q,\sigma,\false) & = \false \\
	  \delta(q,\sigma,\true) & = \true \\
    \delta([\chi\myU\psi,i,j],\sigma,\false) & = (-1,[\chi\myU\psi,i,j])\lor\\*
     & \quad\;((0,[\psi,7])\land(-1,[\myH\chi,i,j]))\\
    \delta([\chi\myU\psi,i,j],\sigma,\true) & = (0,[\psi,7])\\
    \delta([\myH\psi,i,j],\sigma,\false) & = (0,[\psi,7])\land(-1,[\myH\psi,i,j])\\
    \delta([\myH\psi,i,j],\sigma,\true) & = (0,[\psi,7])\\
    \delta([\neg(\chi\myU\psi),i,j],\sigma,\false) & = ((0,[\myH\neg\psi,i,j])\land(\mybox,[\myG\neg\psi,i,j]))\lor\\*
     & \quad\; (-1,[\myP(\neg\chi\land\myH\neg\psi),i,j])\\
    \delta([\neg(\chi\myU\psi),i,j],\sigma,\true) & = (\mybox,[\myG\neg\psi,i,j])\lor((0,[\neg\chi,7])\land(0,[\neg\psi,7]))\\
	  \delta([\myG\psi,i,j],\sigma,b) & = (0,[\psi,7])\land(\mybox,[\myG\psi,i,j]) && \mspace{-68mu}\text{if } p_{j}\in\sigma\\
	  \delta([\myG\psi,i,j],\sigma,b) & = \true && \mspace{-68mu}\text{if } p_{j}\not\in\sigma\\
	  \delta([\myP(\neg\chi\land\myH\neg\psi),i,j],\sigma,\false) & = (0,[\neg\chi,7])\land(0,[\myH\neg\psi,i,j]) \\
	  \delta([\myP(\neg\chi\land\myH\neg\psi),i,j],\sigma,\true) & = (0,[\neg\chi,7])\land(0,[\neg\psi,7]) \\
    \delta([\Finf\psi,i,j],\sigma,b) & = ((\mydiamond,[\Finf\psi,i,j])\land(\mybox,[\Finf\psi,i,j]))\lor\\*
     & \quad\;(0,[(\myX\Finf\psi)\wedge\psi,i,j]) && \mspace{-68mu}\text{if } p_{j}\in\sigma\\
    \delta([\Finf\psi,i,j],\sigma,b) & = \true && \mspace{-68mu}\text{if } p_{j}\not\in\sigma\\
    \delta([\myX\Finf\psi,i,j],\sigma,b) & = (\mydiamond,[\Finf\psi,i,j])\land(\mybox,[\Finf\psi,i,j])\\
    \delta([\neg\Finf\psi,i,j],\sigma,b) & = (0,[\myG\neg\psi,i,j])\\
    \delta([\myY\psi,i,j],\sigma,\false) & = (-1,[\myY\psi,i,j]) \\
	  \delta([\myY\psi,i,j],\sigma,\true) & = (\text{lift},[\myY\psi,8]) \\
    \delta([\neg\myY\psi,i,j],\sigma,\false) & = (-1,[\neg\myY\psi,i,j]) \\
	  \delta([\neg\myY\psi,i,j],\sigma,\true) & = (\text{lift},[\neg\myY\psi,8]) \\
    \delta([\chi\myS\psi,i,j],\sigma,\false) & = (-1,[\chi\myS\psi,i,j])\\
    \delta([\chi\myS\psi,i,j],\sigma,\true) & = (\text{lift},[\chi\myS\psi,8])\\
    \delta([\neg(\chi\myS\psi),i,j],\sigma,\false) & = (-1,[\neg(\chi\myS\psi),i,j])\\
    \delta([\neg(\chi\myS\psi),i,j],\sigma,\true) & = (\text{lift},[\neg(\chi\myS\psi),8])
\end{align*}

The initial state of $\aA_{\varphi'}$ is $[\varphi',1]$. The set $G$ contains all states of the form $[\neg\myE(\chi\myU\psi),i]$, $[\neg\myA(\chi\myU\psi),i]$, and $[(\ex\efinf\psi)\land\psi,i]$ for all $i\in\{1,2,7,8\}$ and the states $[\myG\psi,i,j]$ for $i\in\{4,6\}$ and $j\in[1,2m+2n]$. The set $B$ contains all states of the form $[\neg\ex\efinf\psi,i]$ for all $i\in\{1,2,7,8\}$.

The partition on $Q$ and the order on the resulting sets is defined analogously to Theorem \ref{theo:pectl}. In particular, the states $[(\myX\Finf\psi)\wedge\psi,i,j]$, $[\myX\Finf\psi,i,j]$, and $[\Finf\psi,i,j]$ constitute an existential set for every choice if $i$ and $j$.

To prove correctness of our construction, we observe that $\varphi'$ is satisfiable if and only if $\varphi$ is satisfiable. The proof then proceeds by an induction on the formula structure, considering only the state subformulas. The only non standard argument is the correctness for formulas of the form $\myE\psi_j$. But this is easy to formalize given the explanation above.
\qed
\end{proof}

\section{Proof of Lemma \ref{lem:automata_translation}}

\subsection{Note on the Existence of Homogeneous Runs}
\begin{claim} Let \aA be a \kWPHAA and $T$ a tree. \aA has an accepting run on $T$ if and only if \aA has an accepting homogeneous run on $T$.
\end{claim}
This claim can be proved by observing that the configuration graph of \aA can be seen as the arena of a two-player game with parity winning condition. Note that the acceptance condition of \aA is special kind of parity condition \cite{KupfermanVW00}. The existence of memoryless winning strategies in such games \cite{EmersonJ91,Zielonka98} implies that if \aA has an accepting run (i.e., a winning strategy), then \aA has also a homogeneous accepting run (i.e., a memoryless winning strategy).

\subsection{Construction of \aB}

We describe which information is stored in the states of \aB and how it is used to check existence of an accepting run of \aA. A formal definition of \aB is given afterwards.

First, \aB needs information about which states are taken by \aA during $r$ at $x$. Even more, \aB needs to know which states are taken with respect to the same placement of the pebbles. Fortunately, if two pebble placements $\bar{y}$ and $\bar{y}'$ with $\mpp{\bar{y}}=\mpp{\bar{y}'}$ give rise to the same set $S$ of states occurring at $x$, we do not have to distinguish these sets.

The states of \aB will be of the form $(X_0,\ldots,X_k,\rho)$, where $\rho\in\mathbb{B}$ and each $X_i$ is a set of tuples $(S,uS,dS,rS,uP,dP,auP,adP,rdP,C,aC)\in(2^{Q_{\aA}})^4\times (2^{Q_{\aA}\times Q_{\aA}})^7$, where the sets $S$ are the sets described above and the other components will be explained below. Note that there are at most $2^{O(n^2)}$ many different tuples, where $n=|Q_{\aA}|$. Therefore, there are at most $(2^{2^{n^2}})^{k+1}$ many different states of \aB. As we can assume that $k\leq n$, the size of \aB is at most doubly exponential in $n$.

Let $x$ be a node of $T$ and $(X_0,\ldots,X_k,\rho)$ the state taken by \aB at $x$. $\rho$ will be \true iff $x$ is the root of $T$. The sets $X_i$ will contain the following information about the guessed run $r$: For every pebble placement $\bar{y}$ with $\mpp{\bar{y}}=i$, the set $X_i$ will contain a tuple\footnote{Note that $X_0$ will always contain exactly one tuple as there is only one pebble placement $\bar{y}$ with $\mpp{\bar{y}}=0$.} $(S,uS,dS,rS,uP,dP,auP,adP,rdP,C,aC)$ as described in the following.
\begin{align*}
S & := \{q\in Q_{\aA}\mid (q,x,\bar{y})\text{ occurs in }r\}\\
uS & := \{q\in Q_{\aA}\mid (q,x\cdot -1,\bar{y})\text{ occurs in }r\text{ as child of a node }(q',x,\bar{y})\}\\
dS & := \{q\in Q_{\aA}\mid (q,x,\bar{y})\text{ occurs in }r\text{ as child of a node }(q',x\cdot -1,\bar{y})\}
\end{align*}

The latter two sets are used to check consistency of guesses of \aB. As mentioned above, \aB has to check that the guessed run $r$ is accepting. To this end, we distinguish two kinds of infinite paths in $r$: Those that go down the tree arbitrary far can be handled by the acceptance condition of \aB as described below. But there might be also \emph{bounded} paths, i.e., paths that contain only nodes up to a certain depth. These paths necessarily end up in a loop. \aB has to check that these loops contain a state from $G$ or do not contain a state from $B$, depending on whether they consist of existential or universal states. This will be done at those nodes $x$, such that the loop contains a configuration $(q,x,\bar{y})$ but no configuration with a strict descendant of $x$, except if an additional (compared to $\bar{y}$) pebble was placed before. To perform these checks, \aB maintains information about the following three kinds of paths in $r$.

An \emph{upward path} is a path in $r$ starting from a configuration $(p,x\cdot -1,\bar{y})$ and ending in a configuration $(q,x,\bar{y})$ such that there is no intermediate configuration $(p'x',\bar{y})$ where $x\leq x'$, and there is no intermediate configuration $(p'x',\bar{y}')$ with $\mpp{\bar{y}'}<\mpp{\bar{y}}$. Note that intermediate configurations $(p'x',\bar{y}')$ with $\mpp{\bar{y}'}>\mpp{\bar{y}}$ and $x\leq x'$ are allowed.
$$uP:=\{(p,q)\in (Q_{\aA})^2\mid\text{there is an upward path from }(p,x\cdot -1,\bar{y})\text{ to }(q,x,\bar{y})\}$$

A \emph{downward path} is a path starting from a configuration $(p,x,\bar{y})$ and ending in a configuration $(q,x\cdot -1,\bar{y})$, without an intermediate configuration $(p'x',\bar{y}')$ with $x\not\leq x'$. In particular, the pebbles that are placed at the beginning of a downward path will not be lifted along this downward path as they are placed at $x\cdot -1$ or above.
$$dP:=\{(p,q)\in (Q_{\aA})^2\mid\text{there is a downward path from }(p,x,\bar{y})\text{ to }(q,x\cdot -1,\bar{y})\}$$

As the third kind of paths in $r$, we consider paths that start in a configuration where the last pebble is placed at the current node and that end by lifting this pebble. If either $\mpp{\bar{y}}=0$ or $\mpp{\bar{y}}\geq 1$ and $y_i\neq x$, with $\bar{y}=(y_1,\ldots,y_i,\pnp,\ldots,\pnp)$, then $C=\emptyset$. Otherwise, if $y_i=x$, $C$ is the set of all tuples $(p,q)\in(Q_{\aA})^2$ such that there is a path in $r$ starting from $(p,x,\bar{y})$ and ending at $(q,x,(y_1,\ldots,y_{i-1},\pnp,\ldots,\pnp))$ where $\mpp{\bar{y}'}$ holds for all intermediate configurations $(p',x',\bar{y}')$.

Besides information about the existence of paths of the three kinds described above, \aB also needs information about whether they contain states from $G_{\aA}$ or $B_{\aA}$.
 Let $auP$ be the set of all tuples of states $(p,q)\in (Q_{\aA})^2$ such that there is an upward path in $r$ from $(p,x\cdot -1,\bar{y})$ to $(q,x,\bar{y})$ and 
\begin{itemize}
	\item either $p$ and $q$ belong to the same existential set $Q_j$ and all upward paths from $(p,x\cdot -1,\bar{y})$ to $(q,x,\bar{y})$ contain a state from $G_{\aA}$, or
	\item $p$ and $q$ belong to the same universal set $Q_j$ and there is a upward path from $(p,x\cdot -1,\bar{y})$ to $(q,x,\bar{y})$ that contains a state from $B_{\aA}$.
\end{itemize}
The sets $adP$ and $aC$ are defined analogously.

Whenever \aB guessed some downward path, it has to make sure that this path really exists. But the existence of a downward path at a node $x$ might depend on the existence of other downward paths at the children of $x$. A downward path can thus be seen as a request generating other requests. \aB has to check that this process terminates.

This can be done along with the handling of the first part of acceptance condition for unbounded paths using an interval-technique. For path $\pi$ in $T$ and every node $x$ on $\pi$, we define $w_{\pi}(x)$ to be the minimal node $w_{\pi}(x)>x$ on $\pi$ such that
\begin{enumerate}
	\item for each state $q\in S\cap Q_l$ for some existential set $Q_l$, each subpath of $r$ starting from a node with configuration $q,x,\bar{y}$ and reaching a configuration $(p,w_{\pi}(x),\bar{y})$ with $p\in Q_l$ visits some state from $G$, and
	\item for each $(p,q)\in dP$ there is a downward path in $r$ from $(p,x,\bar{y})$ to $(q,x\cdot -1,\bar{y})$ on which no descendant of $w_{\pi}(x)$ is visited,
\end{enumerate}
where the sets $S$ and $dP$ refer to the tuple corresponding to the pebble placement $\bar{y}$ in the state taken by \aB at $x$.
Now, for every infinite path $\pi$ in $T$, we define an infinite increasing sequence $x_0,x_1,\ldots$ of nodes on $\pi$ as follows: $x_0:=\varepsilon$ and given $x_j$, $x_{j+1}:=w_{\pi}(x_j)$. Using K\"onig's Lemma, we can easily prove that such a sequence exists if and only if $r$ fulfills the first part of the acceptance condition, i.e., the one with respect to the set $G_{\aA}$. The sets $rS$ and $rdP$ are used by \aB to check these conditions. They get initialized at the root and are updated at every transition. As soon as all requests have been fulfilled, they get reinitialized. The states in $G_{\aB}$ will be those where all these sets are empty.

The other part of the acceptance condition, i.e., the one with respect to the universal states and the set $B_{\aA}$, can be checked by \aB using the sets $S$ and $adP$. 

Please note that while the number of states of \aB is doubly exponential in the number of states of \aA, both automata use the same set $D$.

We give the formal definition of $\aB=(Q_{\aB},\Sigma,D,Q^0_{\aB},\delta_{\aB},\{\langle G_{\aB},B_{\aB}\rangle\})$. For sake of presentation, we use a set of initial states $Q^0_{\aB}$ instead of a single state.

The set $Q_{\aB}$ of states of \aB consists of all tuples $(X_0,X_1,\ldots,X_k,\rho)$, where for every $0\leq i\leq k$, $X_i$ is a set of tuples $(S,uS,dS,rS,uP,dP,auP,adP,rdP,C,aC)$ from $(2^{Q_{\aA}})^5\times (2^{Q_{\aA}\times Q_{\aA}})^7$ and $\rho\in\mathbb{B}$, such that
\begin{itemize}
	\item $X_0$ contains exactly one tuple $(S,uS,dS,rS,uP,dP,auP,adP,rdP,C,aC)$ and the sets $dP,adP,rdP,C$ and $aC$ are empty  in this tuple,
	\item for every tuple from every set $X_i$, $i\geq 0$: $dS\subseteq S$, $rS\subseteq S$, $auP\subseteq uP$, $adP\subseteq dP$, and $rdP\subseteq dP$.
\end{itemize}

The initial states of \aB are those states where $\rho=\true$ and
\begin{itemize}
	\item $dS=\{q^0_{\aA}\}$, and therefore $q^0_{\aA}\in S$, for the sets $S,dS$ from the tuple in $X_0$,
	\item for every tuple from every set $X_i$, $i\geq 0$: $uS=\emptyset$, $rS=S$, $uP=\emptyset$, $dS=\emptyset$, and $dP=\emptyset$.
\end{itemize}
Note that this implies that the sets $auP,adP$ and $rdP$ are empty as well. The first condition is used to justify that $q^0_{\aA}\in S$, avoiding a special case in the definition of the transition function below.

For the acceptance condition of \aB, we have to define the two sets $G_{\aB}$ and $B_{\aB}$. The set $G_{\aB}$ consists of all those states, where, for every tuple from every set $X_i$, $rS=\emptyset$ and $rdPS=\emptyset$. In particular, the state where every set $X_i$ contains only one tuple consists only of empty sets is accepting. This state is taken by \aB at all nodes of $T$ not visited by \aA during the run guessed by \aB.
A state is contained in $B_{\aB}$, if for some tuple from some set $X_i$, $S\cap B_{\aA}\neq\emptyset$ or $adP$ contains a tuple $(p,q)$ with $p,q\in Q_j$ for some universal set $Q_j$.

For every $d\in D$, let $c:=(X_0,\ldots,X_k,\rho)$ and $c_j:=(X_0^j,\ldots,X_k^j,\rho_j)$ for $1\leq j\leq d$ be states of \aB. We define the transition function $\delta_{\aB}$ for all $d\in D$ and all $\sigma\in\Sigma$ as follows:

$(c_1,\ldots,c_d)\in\delta_{\aB}(c,\sigma,d):\iff$
\begin{enumerate}[(i)]
	\item for every $i\in[0,k]$ and every $j\in[1,d]$, there is a relation $X_i\sim_{i,j}X_i^j$ such that $\forall t\in X_i\exists t_j\in X_i^j:t\sim_{i,j}t_j$ and $\forall t_j\in X_i^j\exists t\in X_i:t\sim_{i,j}t_j$, 
	\item for every $i\in[1,k]$, there is one tuple $t_i^c\in X_i$, one tuple $l(t_i^c)\in X_{i-1}$, and for every $j\in[1,d]$ one tuple $t_{i,j}^c\in X_i^j$ with $t_i^c\sim_{i,j}t^c_{i,j}$, and
	\item for every tuple $t=(S,\ldots)$ from some set $X_i$ and every state $q\in S$ such that $\delta_{\aA}(q,\sigma,d,b)=\alpha$ for some Boolean combination $\alpha$, where $b=\true$ if and only if $t=t^c_i$, there exists a set $Y_q^t\subseteq(D\cup\{-1,0,\text{root}\})\times Q_{\aA}$ that exactly satisfies $\alpha$,
\end{enumerate}
such that the conditions (1) to (19) given below are satisfied.

The sets $Y_q^t$ represent the guess of \aB on how \aA acts at the current node. Here we use that we only consider homogeneous runs. The relations $\sim_{i,j}$ establish a connection between tuples referring to the same pebble position. Finally, the tuples $t^c$ from (ii) are those referring to the situation where the last placed pebble is placed at the current node and $l(t^c_i)$ refers to the situation reached from the one corresponding to $t^c_i$ by lifting the last pebble. Note that $l(t^c_{i+1})=t^c_i$ if and only if $t^c_{i+1}$ refers to the situation where both pebbles, $i$ and $i+1$, are placed at the current node.

In the following conditions on the transition function, we use $S^c_{i}$ to denote the first set in the tuple $t^c_{i}$ and analogous notation for the other sets from the tuples $t^c_i$ and $t^c_{i,j}$.

The first condition states that after applying a transition, \aB is not at the root of $T$ anymore. This ensures that there is exactly one state $(X_0,\ldots,X_k,\rho)$ with $\rho=\true$ in any run of \aB: the initial state.
\begin{enumerate}
	\item[(1)] For all $j\in[1,d]$: $\rho_j=\false$.
\end{enumerate}

The next six conditions assure that the states $c,c_1,\ldots,c_d$ are consistent with the transition function and that every $c_i$ is consistent with $c$.
\begin{enumerate}
	\item[(2)] For all states $q\in Q_{\aA}$
	\begin{itemize}
		\item if $\delta_{\aA}(q,\sigma,d,b)=(\text{drop},p)$, then
		\begin{itemize}
			\item if $q\in S$ for some tuple $t=(S,\ldots)\in X_i$ for some $i\in[0,k-1]$, then $p\in S^c_{i+1}$ and $t=l(t^c_{i+1})$, and
			\item $q\not\in S$ for all tuples $(S,\ldots)\in X_k$,
		\end{itemize}
		\item if $\delta_{\aA}(q,\sigma,d,b)=(\text{lift},p)$, then
		\begin{itemize}
			\item $q\not\in S$ for all tuples $(S,\ldots)\in X_0$,
			\item if $q\in S^c_i$ for some tuple $t^c_i$ with $i\in[1,k]$, then $p\in l(t^c_i)$, and
			\item for every $i\in[1,k]$ and every $t=(S,\ldots)\in X_i\setminus\{t^c_i\}$, $q\not\in S$,
		\end{itemize}
		\item if $\delta_{\aA}(q,\sigma,d,b)=\alpha$ and $q\in S$ for some tuple $t=(S,uS,\ldots)\in X_i$ with $i\in[0,k]$, then for all $j\in[1,d]$, all tuples $t_j=(\ldots,dS_j,\ldots)\in X_i^j$ with $t\sim_{i,j} t_j$, and all $(e,p)\in Y^t_q$
		\begin{itemize}
			\item if $e=-1$, then $p\in uS$,
			\item if $e=0$, then $p\in S$, 
			\item if $e\geq 1$, then $p\in dS_e$, and
			\item if $e=\text{root}$, then $\rho=\true$ and $p\in S$.
		\end{itemize}
	\end{itemize}
\end{enumerate}
\begin{enumerate}
	\item[(3)] For every $i\in[0,k]$ and every $j\in[1,d]$ and for all tuples $t=(S,\ldots)\in X_i$ and $t_j=(S_j,uS_j,\ldots)\in X_i^j$ with $t\sim_{i,j} t_j$: $uS_j\subseteq S$.
\end{enumerate}
\begin{enumerate}
	\item[(4)] For every state $q\in S$ for some tuple $t=(S,uS,,dS\ldots)\in X_i$ with $i\in[0,k]$, there is a sequence of pairwise disjoint states $q_1,\ldots,q_l$ from $S$, such that 
	\begin{itemize}
	\item $q_l=q$, 
		\item one of the following conditions holds for $q_1$
		\begin{itemize}
			\item $q_1\in dS$,
			\item $q_1\in uS_j$ for some $j\in[1,d]$ and some tuple $t_j=(S_j,uS_j,\ldots)\in X_i^j$ with $t\sim_{i,j}t_j$, or
			\item $i\in[1,k]$, $t=t^c_i$, and there is a state $p\in S'$ with $l(t)=(S',\ldots)\in X_{i-1}$, such that $\delta_{\aA}(p,\sigma,d,b)=(q_1,\text{drop})$, where $b=\true$ if and only if $l(t)=t^c_{i-1}$, or
			\item $i\in[0,k-1]$, $t=l(t^c_{i+1})$, and there is a state $p\in S'$ with $t^c_{i+1}=(S',\ldots)$, such that $\delta_{\aA}(p,\sigma,d,\true)=(q_1,\text{lift})$,
		\end{itemize}
		\item and for every $j$, $1\leq j< l$, $\delta_{\aA}(q_j,\sigma,d,b)=\alpha$ for some Boolean combination $\alpha$ and $(0,q_{j+1})\in Y_{q_j}$, where $b=\true$ if and only if $t=t^c_{i}$.
	\end{itemize}
\end{enumerate}
\begin{enumerate}
	\item[(5)] For all tuples $t=(S,uS,\ldots)\in X_i$ with $i\in[0,k]$ and all states $q\in uS$, there exists a state $p\in S$ such that $(-1,q)\in Y_p^t$.
\end{enumerate}
\begin{enumerate}
	\item[(6)]For every $j\in[1,d]$ and every tuple $t_j=(S_j,uS_j,dS_j,\ldots)\in X_i^j$ with $i\in[0,k]$, there exists a state $p\in S$ such that $(j,q)\in Y_p^t$, where $t=(S,\ldots)\in X_i$ with $t\sim_{i,j}t_j$.
\end{enumerate}

Conditions (7) to (9) check consistency of the information about downward paths. These conditions also propagate request for downward paths. Consistency of the sets $C$ and $aC$ is checked by conditions (10) to (12). 

The conditions (7) to (12) have to hold with respect to every tuple $t=(S,uS,dS,rS,uP,dP,auP,adP,rdP,C,aC)\in X_i$ for all $i\in[0,k]$. Let $t_j=(S_j,uS_j,\ldots,aC_j)\in X_i^j$ for all $j\in[1,d]$ denote those tuples such that $t\sim_{i,j}t_j$.
\begin{enumerate}
	\item[(7)] For all $(p,q)\in dP$, it holds that $p\in S$, $q\in uS$, and there is a sequence of pairwise disjoint states $q_1,\ldots,q_l$, $l\geq 1$, from $S$, such that $q_1=p$, $(-1,q)\in Y_{q_l}^t$ and for each $h\in[1,l-1]$ one of the following conditions holds:
		\begin{itemize}
			\item $(0,q_{h+1})\in Y_{q_h}^t$,
			\item $i<k$ and there is a state $p'\in S'$ with $t^c_{i+1}=(S',\ldots,C',aC')\in X_{i+1}$ and $l(t^c_{i+1})=t$, such that $\delta_{\aA}(q_h,\sigma,d,b)=(\text{drop},p')$ and $(p',q_{h+1})\in C'$, where $b=\true$ if and only if $t=t^c_i$, or
			\item there is a $j\in[1,d]$ and a state $p'\in dS_j$ such that $(j,p')\in Y_{q_h}^t$ and $(p',q_{h+1})\in dP_j$, and if $(p,q)\in rdP$, then $(p',q_{h+1})\in rdP_j$.
		\end{itemize}
\end{enumerate}
\begin{enumerate}
		\item[(8)] For all $(p,q)\in dP$ such that $p$ and $q$ belong to the same universal set: If there is a sequence of states $q_1,\ldots,q_l$ as specified in condition (7), such that one of the states $q_1,\ldots,q_l,q$ is in $B_{\aB}$,  one of the tuples $(p',q_{h+1})\in C'$ is in $aC'$, or one of the tuples $(p',q_{h+1})\in dP_j$ is in $adP_j$, then $(p,q)\in adP$.
\end{enumerate}
\begin{enumerate}
	\item[(9)] For all $(p,q)\in adP$ such that $p$ and $q$ belong to some existential set: $p\in G_{\aA}$, $q\in G_{\aA}$, or the following condition holds: For all sequences of pairwise disjoint states $q_1,\ldots,q_l$ from $S$ such that $q_1=p$ and $(-1,q)\in Y_{q_l}^t$, and
 all sequences $(d_1,p_1),\ldots,(d_{l-1},p_{l-1})$, where $(d_h,p_h)\in (D\cup\{\text{drop},0\})\times Q_{\aA}$ and for all $h\in[1,l]$
		\begin{itemize}
			\item $(d_h,p_h)\in Y_{q_h}^t$, $d_h=0$, and $p_h=q_{h+1}$,
			\item $(d_h,p_h)\in Y_{q_h}^t$, $d_h\geq 1$, and $(p_h,q_{h+1})\in dP_{d_h}$, or
			\item $d_h=\text{drop}$, $p_h\in S'$ with $t^c_{i+1}=(S',\ldots,C',aC')\in X_{i+1}$ and $l(t^c_{i+1})=t$, $\delta(q_h,\sigma,d,b)=(\text{drop},p')$, where $b=\true$ if and only if $t=t^c_i$, and  $(p_h,q_{h+1})\in C'$,
		\end{itemize}
	there is a $h\in[1,l]$ such that
	\begin{itemize}
		\item $q_j\in G_{\aA}$,
		\item $h<l$ and $(p_h,q_{h+1})\in adP_{d_h}$, or
		\item $h<l$ and $(p_h,q_{h+1})\in aC'$.
	\end{itemize}
\end{enumerate}
\begin{enumerate}
	\item[(10)] If $i=0$ or $t\neq t^c$, then $C=\emptyset$. Else, for all $(p,q)\in C$, it holds that $p\in S$, $q\in S'$ with $l(t)=(S',\ldots)$, and there is a sequence of pairwise disjoint states $q_1,\ldots,q_l$, $l\geq 1$, from $S$, such that $q_1=p$, $\delta_{\aA}(q_h,\sigma,d,\true)=(\text{lift},q)$, and for each $h\in[1,l-1]$ one of the following conditions holds:
		\begin{itemize}
			\item $(0,q_{h+1})\in Y_{q_h}^t$,
			\item $(\text{root},q_{h+1})\in Y^t_{q_h}$ and $\rho=\true$,
			\item $i<k$ and there is a state $p'\in S'$ with $t^c_{i+1}=(S',\ldots,C',aC')\in X_{i+1}$ and $l(t^c_{i+1})=t$, such that $\delta_{\aA}(q_h,\sigma,d,\true)=(\text{drop},p')$ and $(p',q_{h+1})\in C'$, 
			\item there is a state $p'\in uS$ such that $(-1,p')\in Y^t_{q_h}$ and $(p',q_{h+1})\in uP$, or
			\item there is a $j\in[1,d]$ and a state $p'\in dS_j$ such that $(j,p')\in Y_{q_h}^t$ and $(p',q_{h+1})\in dP_j$.
		\end{itemize}
\end{enumerate}
\begin{enumerate}
	\item[(11)] For all $(p,q)\in C$ such that $p$ and $q$ belong to the same universal set: If there is a sequence of states $q_1,\ldots,q_l$ as specified in condition (10), such that one of the states $q_1,\ldots,q_l,q$ is in $B_{\aB}$,  one of the tuples $(p',q_{h+1})\in C'$ is in $aC'$, one of the tuples $(p',q_{h+1})\in uP$ is in $auP$, or one of the tuples $(p',q_{h+1})\in dP_j$ is in $adP_j$, then $(p,q)\in aC$.
\end{enumerate}
\begin{enumerate}
	\item[(12)] For all $(p,q)\in aC$ such that $p$ and $q$ belong to some existential set: $p\in G_{\aA}$, $q\in G_{\aA}$, or the following condition holds: For all sequences of pairwise disjoint states $q_1,\ldots,q_l$ from $S$ such that $q_1=p$ and $\delta_{\aA}(q_h,\sigma,d,\true)=\text{lift},q)$, and
 all sequences $(d_1,p_1),\ldots,(d_{l-1},p_{l-1})$, where $(d_h,p_h)\in (D\cup\{\text{drop},0,-1,\text{root}\})\times Q_{\aA}$ and for all $h\in[1,l]$
		\begin{itemize}
			\item $(d_h,p_h)\in Y_{q_h}^t$, $d_h=0$ and $p_h=q_{h+1}$,
			\item $(d_h,p_h)\in Y_{q_h}^t$, $\rho=\true$, $d_h=\text{root}$, and $p_h=q_{h+1}$,
			\item $(d_h,p_h)\in Y_{q_h}^t$, $d_h= -1$ and $(p_h,q_{h+1})\in uP$, 
			\item $(d_h,p_h)\in Y_{q_h}^t$, $d_h\geq 1$ and $(p_h,q_{h+1})\in dP_{d_h}$, or
			\item $d_h=\text{drop}$, $p_h\in S'$ with $t^c_{i+1}=(S',\ldots,C',aC')\in X_{i+1}$ and $l(t^c_{i+1})=t$, $\delta(q_h,\sigma,d,b)=(\text{drop},p')$, where $b=\true$ if and only if $t=t^c_i$, and  $(p_h,q_{h+1})\in C'$,
		\end{itemize}
	there is a $h\in[1,l]$ such that
	\begin{itemize}
		\item $q_j\in G_{\aA}$,
		\item $h<l$ and $(p_h,q_{h+1})\in auP$,
		\item $h<l$ and $(p_h,q_{h+1})\in adP_{d_h}$, or
		\item $h<l$ and $(p_h,q_{h+1})\in aC'$.
	\end{itemize}
\end{enumerate}

The following three conditions check consistency of the information about upward paths. Opposed to the preceding conditions, this information is checked for the states $c_i$ using the information at $c$.
For every $j\in[1,d]$, every $i\in[0,k]$, and every tuple $t_j=(S_j,uS_j,dS_j,rS_j,uP_j,dP_j,auP_j,adP_j,rdP_j,C_j,aC_j)\in X_i^j$, there is a tuple $t=(S,uS,dS,rS,uP,dP,auP,adP,rdP,C,aC)\in X_i$ with $t\sim_{i,j}t_j$, such that the conditions (13) to (15) are satisfied.
\begin{enumerate}
	\item[(13)] For every tuple $(p,q)\in uP_j$, there is a sequence of states $q_1,\ldots,q_l$ from $S$, such that $q_1=p$, $(j,q)\in Y_{q_l}^t$, and for every $h\in[1,l]$
	\begin{itemize}
		\item $(0,q_{h+1})\in Y^t_{q_h}$,
		\item $(\text{root},q_{h+1})\in Y^t_{q_h}$ and $\rho=\true$,
		\item there is a state $p'\in uS$ such that $(-1,p')\in Y^t_{q_h}$ and $(p',q_{h+1})\in uP$, or
		\item $i<k$ and there is a state $p'\in S'$ with $t^c_{i+1}=(S',\ldots,C',aC')\in X_{i+1}$ and $l(t^c_{i+1})=t$, such that $\delta_{\aA}(q_h,\sigma,d,b)=(\text{drop},p')$ and $(p',q_{h+1})\in C'$, where $b=\true$ if and only if $t=t^c_i$.
	\end{itemize}
\end{enumerate}
\begin{enumerate}
	\item[(14)] For all $(p,q)\in uP_j$ such that $p$ and $q$ belong to the same universal set: If there is a sequence of states $q_1,\ldots,q_l$ as specified in condition (13), such that one of the states $q_1,\ldots,q_l,q$ is in $B_{\aB}$,  one of the tuples $(p',q_{h+1})\in C'$ is in $aC'$, or one of the tuples $(p',q_{h+1})\in uP$ is in $auP$, then $(p,q)\in auP_j$.
\end{enumerate}
	
\begin{enumerate}
	\item[(15)] For all $(p,q)\in auP_j$ such that $p$ and $q$ belong to some existential set: $p\in G_{\aA}$, $q\in G_{\aA}$, or the following condition holds: For all sequences of pairwise disjoint states $q_1,\ldots,q_l$ from $S$ such that $q_1=p$ and $(j,q)\in Y_{q_l}^t$, and
 all sequences $(d_1,p_1),\ldots,(d_{l-1},p_{l-1})$, where $(d_h,p_h)\in \{\text{drop},0,-1,\text{root}\}\times Q_{\aA}$ and for all $h\in[1,l]$
		\begin{itemize}
			\item $(d_h,p_h)\in Y_{q_h}^t$, $d_h=0$, and $p_h=q_{h+1}$,
			\item $(d_h,p_h)\in Y_{q_h}^t$, $\rho=\true$, $d_h=\text{root}$, and $p_h=q_{h+1}$,
			\item $(d_h,p_h)\in Y_{q_h}^t$, $d_h= -1$, and $(p_h,q_{h+1})\in uP$, or
			\item $d_h=\text{drop}$, $p_h\in S'$ with $t^c_{i+1}=(S',\ldots,C',aC')\in X_{i+1}$ and $l(t^c_{i+1})=t$, $\delta(q_h,\sigma,d,b)=(\text{drop},p')$, where $b=\true$ if and only if $t=t^c_i$, and  $(p_h,q_{h+1})\in C'$,
		\end{itemize}
	there is a $h\in[1,l]$ such that
	\begin{itemize}
		\item $q_j\in G_{\aA}$,
		\item $h<l$ and $(p_h,q_{h+1})\in auP$, or
		\item $h<l$ and $(p_h,q_{h+1})\in aC'$.
	\end{itemize}
\end{enumerate}
	
Conditions (16) and (17) are used to check that every bounded path in $r$ is accepting. As described above, we check that every \emph{upward loop} contains a state $G_{\aA}$ if all states in the loop are existential, and that it does not contain a state from $B_{\aA}$ if they are universal.
\begin{enumerate}
	\item[(16)] For every tuple $t=(S,uS,dS,rS,uP,dP,auP,adP,rdP,C,aC)$ with $t\in X_i$ for some $i\in[0,k]$, and every existential state $q\in S$: If there is a sequence of states $q_1,\ldots,q_l$ from $S$ with $q_1=q_l=q$ and a sequence of tuples $(d_1,p_1),\ldots,(d_{l-1},p_{l-1})$ with $(d_h,p_h)\in \{\text{drop},0,-1,\text{root}\}\times Q_{\aA}$, such that for all $h\in[1,l]$
		\begin{itemize}
			\item $(d_h,p_h)\in Y_{q_h}^t$, $d_h=0$ and $p_h=q_{h+1}$,
			\item $(d_h,p_h)\in Y_{q_h}^t$, $\rho=\true$, $d_h=\text{root}$, and $p_h=q_{h+1}$,
			\item $(d_h,p_h)\in Y_{q_h}^t$, $d_h= -1$ and $(p_h,q_{h+1})\in uP$, or
			\item $d_h=\text{drop}$, $p_h\in S'$ with $t^c_{i+1}=(S',\ldots,C',aC')\in X_{i+1}$ and $l(t^c_{i+1})=t$, $\delta(q_h,\sigma,d,b)=(p',\text{drop})$, where $b=\true$ if and only if $t=t^c_i$, and  $(p_h,q_{h+1})\in C'$,
		\end{itemize}
	then there is a $h\in[1,l]$ such that
	\begin{itemize}
		\item $q_j\in G_{\aA}$,
		\item $h<l$ and $(p_h,q_{h+1})\in auP$, or
		\item $h<l$ and $(p_h,q_{h+1})\in aC'$.
	\end{itemize}
\end{enumerate}
\begin{enumerate}
	\item[(17)] For every tuple $t=(S,uS,dS,rS,uP,dP,auP,adP,rdP,C,aC)$ with $t\in X_i$ for some $i\in[0,k]$, and every universal state $q\in S$: For all sequences of states $q_1,\ldots,q_l$ from $S$ with $q_1=q_l=q$ and
 all sequences $(d_1,p_1),\ldots,(d_{l-1},p_{l-1})$, where $(d_h,p_h)\in \{\text{drop},0,-1,\text{root}\}\times Q_{\aA}$ and for all $h\in[1,l]$
		\begin{itemize}
			\item $(d_h,p_h)\in Y_{q_h}^t$, $d_h=0$ and $p_h=q_{h+1}$,
			\item $(d_h,p_h)\in Y_{q_h}^t$, $\rho=\true$, $d_h=\text{root}$, and $p_h=q_{h+1}$,
			\item $(d_h,p_h)\in Y_{q_h}^t$, $d_h= -1$ and $(p_h,q_{h+1})\in uP$, or
			\item $d_h=\text{drop}$, $p_h\in S'$ with $t^c_{i+1}=(S',\ldots,C',aC')\in X_{i+1}$ and $l(t^c_{i+1})=t$, $\delta(q_h,\sigma,d,b)=(p',\text{drop})$, where $b=\true$ if and only if $t=t^c_i$, and  $(p_h,q_{h+1})\in C'$,
		\end{itemize}
	the following conditions hold for every $h\in[1,l]$:
	\begin{itemize}
		\item $q_j\not\in B_{\aA}$,
		\item if $h<l$, then $(p_h,q_{h+1})\not\in auP$, and
		\item if $h<l$, then $(p_h,q_{h+1})\not\in aC'$.
	\end{itemize}
\end{enumerate}

The next condition updates the \emph{request sets} used to check the acceptance of unbounded paths in $r$. This condition has to hold with respect to every tuple $t=(S,uS,dS,rS,uP,dP,auP,adP,rdP,C,aC)\in X_i$ for all $i\in[0,k]$. Let $t_j=(S_j,uS_j,dS_j,rS_j\ldots,aC_j)\in X_i^j$ for all $j\in[1,d]$ denote those tuples such that $t\sim_{i,j}t_j$.
\begin{enumerate}
	\item[(18)] For all existential states $p\in rS$, all $j\in[1,d]$, and all states $q\in dS_j$, if there is a sequence of pairwise disjoint states $q_1,\ldots,q_l$ from $S$ such that
	\begin{itemize}
		\item $q_1=p$, 
		\item $(j,q)\in Y_{q_l}^t$,
		\item $q\not\in G_{\aA}$, $q_h\not\in G_{\aA}$ for $h\in[1,l]$, and
		\item for all $h\in[1,l-1]$,
		\begin{itemize}
			\item $(0,q_{h+1})\in Y_{q_j}^t$,
			\item $(\text{root},q_{h+1})\in Y^t_{q_h}$ and $\rho=\true$,
			\item there is a state $p'\in uS$ such that $(-1,p')\in Y_{q_h}^t$ and $(p',q_{h+1})\in uP\setminus auP$, or
		\item $i<k$ and there is a state $p'\in S'$ with $t^c_{i+1}=(S',\ldots,C',aC')\in X_{i+1}$ and $l(t^c_{i+1})=t$, such that $\delta_{\aA}(q_h,\sigma,d,b)=(\text{drop},p')$ and $(p',q_{h+1})\in C'\setminus aC'$, where $b=\true$ if and only if $t=t^c_i$,
		\end{itemize}
	\end{itemize}
	then $q\in rS_j$.
\end{enumerate}

The last condition reinitializes the \emph{request sets} after all previous requests have been fulfilled.
\begin{enumerate}
	\item[(19)] If, for every tuple $(S,uS,dS,rS,uP,dP,auP,adP,rdP,C,aC)\in X_i$ with $i\in[0,k]$, $rS=\emptyset$ and $rdP=\emptyset$, then for all $j\in[1,d]$, all $i\in[0,k]$, and all tuples $(S_j,uS_j,dS_j,rS_j,uP_j,dP_j,auP_j,adP_j,rdP_j,C_j,aC_j)\in X_i^j$
	\begin{itemize}
		\item $rS_j=S_j\setminus G_{\aA}$ and
		\item $rdP_j=dP_j$.
	\end{itemize}
\end{enumerate}
This concludes our construction.

To prove correctness of our construction, assume that $r$ is an accepting run of \aA on a tree $T$. Obviously, we get an accepting run for \aB on $T$ by choosing those states that represent the run $r$ as indicated above.

For the reverse direction, assume that there is an accepting run of \aB on $T$. For every node $x$ of $T$, there is a possible behavior of \aA at $x$ (the sets $Y_q^t$ above) that is consistent with the run of \aB. From these, starting with the initial state of \aA at the root of $T$, we can construct a run for \aA and show that this run is accepting. The details are basically a reproduction of the construction given above.

\section{Proof of Proposition \ref{prop:bounded_branching}}
Let $T\in L(\aA)$ be a tree and $r$ a homogeneous accepting run of \aA on $T$. We show how to obtain from $T$ a tree $T'$, such that the branching degree of $T'$ is at most $2^{n^2\cdot(k+1)}$ and $T'\in L(\aA)$.

Let $x$ be some node of $T$ and $T_x$ the subtree rooted at $x$. Assume that \aA enters $T_x$ in a state $q$ during the run $r$ while having placed $i$ pebbles outside of $T_x$. I.e., we consider a configuration $c:=(q,x,\bar{y})$ with $\mpp{\bar{y}}=i$ and $y_{\mpp{\bar{y}}}\neq x$. Every path starting from $c$ either stays in $T_x$ or leaves $T_x$ in taking some state $p$ at the parent node of $x$. Put into the notation used in the proof of Lemma \ref{lem:automata_translation}: There is a downward path from $q$ to $p$. 

Let $f_i^x:Q\rightarrow 2^Q$ be the mapping assigning to every state $q$ the set of states in which $T_x$ is left when entered in $q$ with $k-i$ pebbles in hand.
Note that $T_x$ might be entered many times during $r$. But since $r$ is homogeneous, the resulting sets of states are always the same, i.e., the functions $f_i^x$ for $i\in[0,k]$ are well defined and describe how \aA behaves on $T_x$. We therefore call $f_0^x,\ldots,f_k^x$ the type of $x$.
As the size of a type is $n^2\cdot(k+1)$, there are at most $2^{n^2\cdot(k+1)}$ different types.

Now assume that there are two sibling nodes $x$ and $y$ in $T$ that have the same type. Let $T'$ be the tree obtained from $T$ by removing $T_y$. Then we can get an accepting homogeneous run $r'$ of \aA on $T'$ by moving into $T_x$ whenever going into $T_y$ in $r$. Note that this refers to $\mydiamond$-moves and uses that $r$ is accepting, i.e., that every infinite path of $r$ in $T_x$ and $T_y$ is accepting.

Applying this procedure level-wise top-down results in a tree with branching degree bounded by $2^{n^2\cdot(k+1)}$ that is accepted by \aA.

\end{document}